\documentclass[journal]{IEEEtran}
\usepackage{booktabs} 
\usepackage{makecell}
\usepackage{graphicx}
\usepackage{amsmath, amsthm, amssymb}
\usepackage{algorithm}
\usepackage{algpseudocode}
\usepackage{multirow}
\usepackage{hyperref} 
\usepackage{lineno, url} 
\usepackage{siunitx}

\newtheorem{theorem}{Theorem}

\newtheorem{assumption}{Assumption}

\hypersetup{
    colorlinks=true,
    linkcolor=blue,
    filecolor=magenta,
    urlcolor=cyan,
    citecolor=red, 
}

\begin{document}

\title{PQS-BFL: A Post-Quantum Secure Blockchain-based Federated Learning Framework}

\author{Daniel~Commey
        and~Garth~V.~Crosby%
\thanks{D. Commey is with the Department of Multidisciplinary Engineering, Texas A\&M University, College Station, TX 77843 USA (e-mail: dcommey@tamu.edu).}%
\thanks{G. V. Crosby is with the Department of Engineering Technology \& Industrial Distribution, Texas A\&M University, College Station, TX 77843 USA (e-mail: gvcrosby@tamu.edu).}}%

\maketitle

\begin{abstract}
Federated Learning (FL) enables collaborative model training while preserving data privacy, but its classical cryptographic underpinnings are vulnerable to quantum attacks. This vulnerability is particularly critical in sensitive domains like healthcare. This paper introduces PQS-BFL (Post-Quantum Secure Blockchain-based Federated Learning), a framework integrating post-quantum cryptography (PQC) with blockchain verification to secure FL against quantum adversaries. We employ ML-DSA-65 (a FIPS 204 standard candidate, formerly Dilithium) signatures to authenticate model updates and leverage optimized smart contracts for decentralized validation. Extensive evaluations on diverse datasets (MNIST, SVHN, HAR) demonstrate that PQS-BFL achieves efficient cryptographic operations (average PQC sign time: \SI{0.65}{ms}, verify time: \SI{0.53}{ms}) with a fixed signature size of \SI{3309}{Bytes}. Blockchain integration incurs a manageable overhead, with average transaction times around \SI{4.8}{s} and gas usage per update averaging \SI{1.72e6}{} units for PQC configurations. Crucially, the cryptographic overhead relative to transaction time remains minimal (around 0.01-0.02\% for PQC with blockchain), confirming that PQC performance is not the bottleneck in blockchain-based FL. The system maintains competitive model accuracy (e.g., over 98.8\% for MNIST with PQC) and scales effectively, with round times showing sublinear growth with increasing client numbers. Our open-source implementation and reproducible benchmarks validate the feasibility of deploying long-term, quantum-resistant security in practical FL systems.
\end{abstract}


\IEEEpeerreviewmaketitle

\section{Introduction}
\label{sec:introduction}

Federated learning (FL) has emerged as a pivotal paradigm for collaborative machine learning, enabling multiple parties to train models on decentralized data without exposing raw information \cite{mcmahan_communication-efficient_2017}. This privacy-preserving approach holds immense potential, particularly in domains handling sensitive data, such as healthcare \cite{rieke_future_2020, kaissis_secure_2020}. However, the security foundations of conventional FL systems rely heavily on classical cryptographic algorithms like RSA and ECDSA for tasks such as participant authentication and update integrity verification \cite{bonawitz_practical_2017}.

The rapid advancement of quantum computing poses a significant threat to these classical cryptographic primitives. Shor's algorithm, for instance, demonstrates that large-scale quantum computers can efficiently break the mathematical problems underlying RSA and ECDSA \cite{shor_polynomial-time_1999, preskill_quantum_2019}. This impending "quantum threat" necessitates a transition towards post-quantum cryptography (PQC)—cryptographic systems designed to resist attacks from both classical and quantum computers \cite{bernstein_post-quantum_2017}. The National Institute of Standards and Technology (NIST) is actively standardizing PQC algorithms, with lattice-based schemes like ML-DSA (formerly Dilithium) emerging as leading candidates for digital signatures \cite{alagic_status_2022}.

Simultaneously, blockchain technology has been explored as a means to enhance the security, transparency, and auditability of FL systems \cite{bonawitz_practical_2017, harris_decentralized_2019}. By recording model updates or their hashes on an immutable ledger, blockchain can provide verifiable proof of contributions and prevent malicious tampering. However, most existing blockchain platforms also rely on classical cryptography, inheriting the same vulnerabilities to quantum attacks.

These simultaneous advances and emerging threats motivate a unified approach that secures FL ecosystems against future quantum adversaries while leveraging the benefits of blockchain verification.

\subsection{Motivation}
The long-term viability of FL, especially in sectors like healthcare where data sensitivity and regulatory compliance (e.g., HIPAA, GDPR) demand enduring security, hinges on adopting quantum-resistant measures. Relying on classical cryptography creates a future vulnerability: data encrypted or systems secured today could be compromised retrospectively once powerful quantum computers become available. Integrating PQC, specifically quantum-resistant digital signatures, is crucial for ensuring the authenticity and integrity of model updates exchanged within an FL system over extended periods. Furthermore, combining PQC with blockchain provides a robust, decentralized mechanism for verifying these updates immutably, fostering trust among participants, and avoiding a central authority that could become a single point of failure.

\subsection{Challenges}
Designing and implementing a Post-Quantum Secure Blockchain-based Federated Learning (PQS-BFL) framework presents several key challenges:
\begin{itemize}
    \item \textbf{PQC Performance Overhead:} PQC algorithms, particularly lattice-based signatures, often have larger key sizes, signature sizes, and potentially higher computational costs compared to their classical counterparts \cite{alagic_status_2022}. Integrating these into FL must not impose prohibitive performance penalties.
    \item \textbf{Blockchain Integration Costs:} Storing PQC signatures (which are larger) and performing PQC verification within smart contracts can lead to increased gas costs and transaction latency on blockchain platforms \cite{zhang_blockchain_2021}. Efficient smart contract design is critical.
    \item \textbf{Scalability:} The system must scale effectively to accommodate a potentially large number of participating clients, ensuring that cryptographic and blockchain operations do not become bottlenecks as the system grows \cite{bonawitz_towards_2019}.
    \item \textbf{Verification Efficiency:} Verifying potentially numerous PQC signatures submitted by clients each round needs to be efficient to maintain reasonable round times.
\end{itemize}

\subsection{Our Contributions}
This paper introduces PQS-BFL, a framework designed to address these challenges. Our primary contributions are:
\begin{enumerate}
    \item \textbf{Quantum-Resistant Signature Integration:} We successfully integrate the ML-DSA-65 signature scheme into an FL framework for authenticating model updates, providing security against known quantum attacks.
    \item \textbf{Blockchain-Based PQC Verification:} We design and implement a blockchain system with smart contracts capable of verifying ML-DSA-65 signatures, enabling decentralized and immutable validation of FL contributions.
    \item \textbf{Comprehensive Empirical Evaluation:} We conduct extensive experiments comparing PQS-BFL (using ML-DSA-65) against ECDSA and a no-cryptography baseline across multiple datasets (MNIST, SVHN, HAR) and client scales (3, 10, 30), analyzing model accuracy, cryptographic performance, blockchain overhead (gas, transaction time), and overall system scalability.
    \item \textbf{Overhead Analysis:} We quantify the overhead introduced by PQC and blockchain, demonstrating that cryptographic overhead is minimal ($<$0.1\%) within the context of blockchain transaction times.
    \item \textbf{Open-Source Framework:} We release our implementation (\texttt{PQS-BFL}) and results to facilitate reproducibility and further research (\url{https://github.com/dcommey/PQS-BFL}).
\end{enumerate}

\subsection{Practical Impact}
PQS-BFL demonstrates the practical feasibility of building quantum-resistant FL systems. Our results show that the performance overhead associated with current PQC standards (ML-DSA-65) is manageable, especially when deployed within a blockchain context. This provides a pathway for organizations, particularly in security-critical sectors like healthcare, to proactively upgrade their collaborative ML infrastructures for long-term, quantum-resilient security without sacrificing operational efficiency or model performance.

\subsection{Paper Organization}
The remainder of this paper is organized as follows. Section~\ref{sec:related} reviews related work. Section~\ref{sec:background} provides background on FL, PQC, and blockchain concepts. Section~\ref{sec:design} details the PQS-BFL system architecture and protocol. Section~\ref{sec:setup} describes the experimental setup. Section~\ref{sec:results} presents and discusses the experimental results. Section~\ref{sec:limitations} discusses limitations and future work. Section~\ref{sec:conclusion} concludes the paper.

\section{Related Work}
\label{sec:related}

Our work on PQS-BFL integrates advancements from several fields: the security of federated learning (FL), the development of post-quantum cryptography (PQC), the application of blockchain technology to machine learning (ML), and performance considerations like gas optimization. This section reviews key contributions in these areas and identifies the gap addressed by our research.

\subsection{Federated Learning Security}
Initial FL frameworks, like FedAvg \cite{mcmahan_communication-efficient_2017}, prioritized communication efficiency and basic privacy by keeping data local. Subsequent research focused on enhancing security against various threats. Secure aggregation protocols \cite{bonawitz_practical_2017} were developed using classical cryptographic techniques (like threshold cryptography or multiparty computation) to prevent the central server from inferring individual client updates during aggregation. Differential privacy \cite{geyer_differentially_2018, agarwal_cpsgd_2018} adds calibrated noise to updates to provide formal privacy guarantees against inference attacks. Homomorphic encryption \cite{gursoy_privacy-preserving_2016, zhang_batchcrypt_2020} allows computations (like aggregation) on encrypted data, further protecting update confidentiality. Byzantine resilience techniques \cite{blanchard_machine_2017, yin_byzantine-robust_2018} aim to detect or mitigate the impact of malicious clients submitting faulty or poisonous updates designed to disrupt training. While these methods address important classical security and privacy concerns, they predominantly rely on cryptographic primitives (e.g., Diffie-Hellman, RSA, ECC) that are known to be vulnerable to attacks by sufficiently powerful quantum computers \cite{shor_polynomial-time_1999}.

\subsection{Post-Quantum Cryptography (PQC)}
The anticipation of quantum computing capabilities has driven the development and standardization of PQC algorithms resistant to quantum attacks \cite{bernstein_post-quantum_2017}. The NIST PQC standardization process \cite{alagic_status_2022, alagic_status_2019} has been pivotal, identifying promising candidates based on different mathematical foundations, including lattices, codes, hashes, and multivariate equations \cite{stewart_committing_2018}. Lattice-based cryptography \cite{regev_lattices_2009, peikert_decade_2016}, built on the presumed hardness of problems like LWE and SIS, is particularly prominent \cite{ducas_crystals-dilithium_2018}. ML-DSA (Dilithium) \cite{kiltz_concrete_2018, lyubashevsky_lattice-based_2021}, selected by NIST for signature standardization, is a key example used in our work. Hash-based signatures like SPHINCS+ (also selected by NIST) and XMSS offer security based only on hash function properties but can have limitations like statefulness or larger signature sizes \cite{bernstein_sphincs_2015, buchmann_security_2011}. Research is ongoing to evaluate PQC performance \cite{paquin_benchmarking_2020, basu_nist_2019} and integrate these schemes into existing protocols \cite{bernstein_post-quantum_2017, bindel_transitioning_2017}. Applying PQC to secure distributed ML systems like FL is an emerging necessity. Some works have started exploring PQC for specific FL aspects, such as secure aggregation using lattice-based homomorphic encryption \cite{xu_laf_2022} or key exchange \cite{ding_post-quantum_2024}, but end-to-end frameworks combining PQC signatures with blockchain verification are less common.

\subsection{PQC Algorithm Selection for FL Systems}
\label{subsec:pqc_selection}

The selection of an appropriate PQC algorithm for integration with blockchain-based FL requires careful consideration of performance characteristics, security levels, and standardization status. As part of our research methodology, we conducted an extensive comparative analysis of three leading NIST PQC digital signature finalists: ML-DSA-65 (formerly Dilithium), Falcon-512, and SPHINCS+-SHA2-128s \cite{commey_post-quantum_2025}. Table~\ref{tab:pqc_comparison} presents the key performance metrics of these algorithms based on our benchmarking using the liboqs library.

\begin{table}[!htbp]
\caption{Performance Comparison of Selected PQC Signature Schemes}
\label{tab:pqc_comparison}
\centering
\sisetup{round-mode=places, round-precision=3}
\resizebox{\linewidth}{!}{%
\begin{tabular}{l S[round-precision=3, table-format=2.3] 
                  S[round-precision=3, table-format=3.3] 
                  S[round-precision=3, table-format=1.3] 
                  S[round-precision=0, table-format=4.0]}
\toprule
{Scheme} & {Key Gen (ms)} & {Sign (ms)} & {Verify (ms)} & {Sig Size (B)} \\
\midrule
ML-DSA-65 & 0.142 & 0.656 & 0.536 & 3309 \\
Falcon-512 & 5.414 & 3.282 & 0.301 & 666 \\
SPHINCS+-SHA2-128s & 17.403 & 131.926 & 3.637 & 7856 \\
\bottomrule
\end{tabular}%
}
\end{table}

\begin{table}[!htbp]
\caption{Security Analysis of Selected PQC Signature Schemes}
\label{tab:pqc_security}
\centering
\resizebox{\linewidth}{!}{%
\begin{tabular}{l >{\centering\arraybackslash}p{1.8cm} >{\centering\arraybackslash}p{1.8cm} >{\centering\arraybackslash}p{1.5cm} p{3.5cm}}
\toprule
{Scheme} & {Classical Security} & {Quantum Security} & {NIST Level} & {Security Foundation} \\
\midrule
ML-DSA-65 & 125 bits & 64 bits & Level 2 & MLWE and MSIS problems \\
Falcon-512 & 118 bits & 58 bits & Level 1+ & NTRU and SIS problems \\
SPHINCS+-SHA2-128s & 133 bits & 66 bits & Level 2 & Hash function properties \\
\bottomrule
\end{tabular}%
}
\end{table}

Our benchmarking revealed key performance characteristics that influenced our decision:

\begin{itemize}
    \item \textbf{Verification Time:} In blockchain-based FL, signature verification occurs frequently and directly impacts transaction confirmation times. ML-DSA-65 (0.536 ms) and Falcon-512 (0.301 ms) offer significantly faster verification than SPHINCS+-SHA2-128s (3.637 ms).
    
    \item \textbf{Signing Time:} Clients sign updates every training round. ML-DSA-65 (0.656 ms) provides the fastest signing operation, considerably outperforming Falcon-512 (3.282 ms) and especially SPHINCS+-SHA2-128s (131.926 ms), whose signing time would introduce noticeable delays in FL rounds.
    
    \item \textbf{Signature Size:} Larger signatures increase blockchain transaction sizes and gas costs. Falcon-512 produces the smallest signatures (666 B), with ML-DSA-65 (3309 B) offering a middle ground, and SPHINCS+-SHA2-128s requiring substantially larger payloads (7856 B).
    
    \item \textbf{Security Level:} As shown in Table~\ref{tab:pqc_security}, all three schemes provide adequate security against quantum attacks, with SPHINCS+-SHA2-128s offering slightly stronger guarantees but at significant performance cost.
    
    \item \textbf{Standardization Status:} ML-DSA-65 has been selected as a FIPS 204 standard, providing greater assurance of long-term support and optimization compared to the other candidates.
\end{itemize}

\begin{figure}[!htbp]
    \centering
    \includegraphics[width=\linewidth]{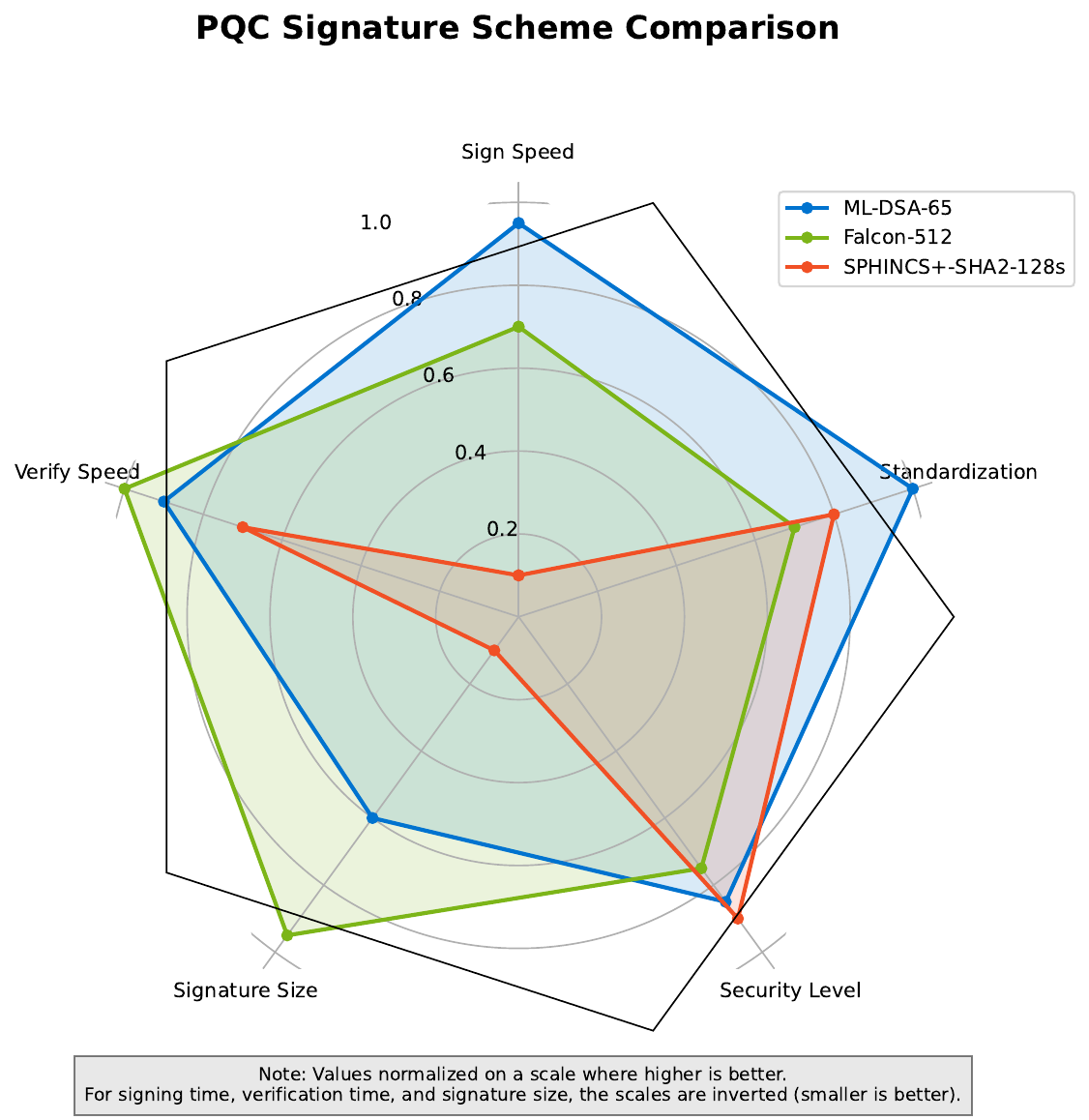}
    \caption{Radar chart comparing the relative strengths of PQC signature schemes across key metrics. Values are normalized on a logarithmic scale where appropriate, with higher scores indicating better performance.}
    \label{fig:pqc_radar}
\end{figure}

When evaluating these schemes for blockchain-based FL integration, we also assessed their gas consumption in smart contract verification. Our preliminary testing indicated that ML-DSA-65 achieves the best balance of gas efficiency and standardization status, though Falcon-512 offers advantages in signature size.

Based on this analysis, we selected ML-DSA-65 for PQS-BFL implementation due to its balanced profile: fast signing and verification, manageable signature size, strong security guarantees, and NIST standardization status. Figure~\ref{fig:pqc_radar} visualizes these trade-offs across the evaluated schemes.

\subsection{Blockchain-Based Machine Learning}
Blockchain technology offers properties like decentralization, immutability, transparency, and auditability, which can benefit ML systems. Researchers have proposed using blockchain for creating decentralized data marketplaces \cite{ranganthan_decentralized_2018}, providing verifiable audit trails for model training and updates \cite{kim_blockchained_2019, ramanan_baffle_2020}, implementing incentive mechanisms for participation in collaborative learning \cite{kang_incentive_2019, huang_collaboration_2024}, and enhancing the security and trust of FL \cite{harris_decentralized_2019, ma_blockchain-based_2020, qammar_securing_2023, arachchige_analysis_2024}, with applications explored in diverse domains such as autonomous vehicles \cite{pokhrel_federated_2020} and Industry 4.0 \cite{qu_blockchained_2020}. Specifically for FL, blockchain can act as a decentralized registry for participants, a platform for verifying update integrity via smart contracts, or a mechanism for coordinating the aggregation process without a fully trusted central server \cite{weng_deepchain_2019}. However, many existing blockchain-FL integrations rely on the blockchain's native classical cryptography, thus inheriting quantum vulnerabilities.

\subsection{Quantum-Resistant Blockchain Systems}
Efforts are underway to secure blockchain platforms themselves against quantum threats. This involves replacing classical signature schemes (like ECDSA) with PQC alternatives for signing transactions \cite{ding_post-quantum_2024, chalkias_blockchained_2018, stewart_committing_2018}. Platforms like Quantum Resistant Ledger (QRL) \cite{waterland_quantum_2016} utilize hash-based signatures (XMSS) for inherent quantum resistance. Research also explores quantum-resistant consensus mechanisms \cite{zhang_blockchain_2021} and the design of smart contracts that can handle PQC primitives \cite{sun_logic_2021}. Integrating these concepts specifically within the context of verifying FL updates signed with PQC signatures is a necessary step towards fully quantum-resistant decentralized learning.

\subsection{Research Gap and Positioning}
While the fields above have seen significant progress, a gap exists in the practical integration and comprehensive evaluation of PQC, specifically quantum-resistant \textit{signatures} for update authentication, within a \textit{blockchain-verified} \textit{federated learning} system. Some recent works have started exploring this intersection, tackling aspects like hybrid PQC schemes \cite{gurung_quantum_2023, gurung_performance_2023, gurung_secure_2023}, secure communication protocols \cite{gharavi_pqbfl_2025, gharavi_post-quantum_2024}, privacy preservation \cite{ding_post-quantum_2024}, or applying PQC-BFL concepts to specific environments like Mobile Edge Computing \cite{xu_post_2023}. Key examples include:
\begin{itemize}
    \item \textbf{Gurung et al.} \cite{gurung_quantum_2023, gurung_performance_2023, gurung_secure_2023} have proposed several frameworks combining PQC (including Dilithium and hybrid approaches) with blockchain-based FL (BFL), focusing on aspects like device role selection, secure communication models, consensus mechanisms (VRF), and MEC integration. Their work often explores hybrid PQC schemes and performance analysis.
    \item \textbf{Gharavi et al.} propose PQBFL \cite{gharavi_pqbfl_2025, gharavi_post-quantum_2024}, a protocol using PQC (Kyber KEM, possibly others) and blockchain for FL. Their focus includes using ratcheting mechanisms for forward secrecy and post-compromise security across FL rounds, hybrid on-chain/off-chain communication, and reputation management. Their work emphasizes session security more than the signature verification aspect central to PQS-BFL.
    \item \textbf{Ding et al.} \cite{ding_post-quantum_2024} focus on privacy preservation in FL using PQC (Kyber KEM) combined with SMPC techniques like secret sharing and ZKP to prevent gradient leakage, rather than focusing on blockchain verification of signed updates.
    \item Other works mention the \textit{need} for quantum-resistant blockchain for AI/FL in healthcare \cite{rahman_blockchain_2019} or discuss quantum encryption for FL security \cite{chu_cryptoqfl_2023, sorbera_adaptive_2025}, but may not provide specific integrated protocol designs or evaluations combining PQC signatures, FL, and blockchain verification.
\end{itemize}

Our work, PQS-BFL, distinguishes itself by:
\begin{enumerate}
    \item Focusing specifically on integrating a standardized PQC \textit{signature} scheme (ML-DSA-65) for authenticating FL updates.
    \item Implementing and evaluating the performance impact of this integration within a \textit{blockchain verification} framework.
    \item Providing a comprehensive empirical comparison against classical ECDSA and a baseline, analyzing cryptographic performance, FL accuracy, blockchain overhead (gas, transaction time), cryptographic overhead ratio, and scalability across multiple datasets.
    \item Demonstrating the practical feasibility and quantifying the performance trade-offs (especially the low relative overhead of PQC in BC context) for this specific architecture.
\end{enumerate}

Table~\ref{tab:related_work_comparison} provides a comparative overview.

\begin{table*}[!ht]
\caption{Comparison with Related Studies Combining PQC/Blockchain/FL Aspects}
\label{tab:related_work_comparison}
\centering
\footnotesize
\resizebox{\textwidth}{!}{%
\begin{tabular}{@{}l@{} >{\raggedright\arraybackslash}p{3cm} >{\raggedright\arraybackslash}p{2.5cm} >{\raggedright\arraybackslash}p{3cm} >{\raggedright\arraybackslash}p{1.5cm} >{\raggedright\arraybackslash}p{4cm}@{}}
\toprule
\textbf{Study} & \textbf{Focus Area} & \textbf{PQC Integration} & \textbf{Blockchain Use} & \textbf{Federated Learning} & \textbf{Key Differentiator / Relation to PQS-BFL} \\
\midrule
\textbf{This Work (PQS-BFL)} & End-to-end framework for quantum-resistant FL update verification using blockchain. Performance evaluation. & Yes (ML-DSA-65 signatures for update authentication) & Yes (Smart contract for PQC signature verification, recording hashes) & Yes (Core application) & Focus on PQC signature verification on-chain, detailed performance analysis (crypto, FL, BC, overhead), open implementation. \\
\addlinespace
Gurung et al. \cite{gurung_quantum_2023, gurung_performance_2023, gurung_secure_2023} & Hybrid PQC for BFL, consensus mechanisms, role selection, MEC integration. & Yes (Hybrid: Dilithium/Falcon + XMSS proposed; also mentions QFL) & Yes (Consensus, BFL framework) & Yes & Explores hybrid PQC, VRF for consensus, broader system aspects. Less focus on detailed benchmark of just PQC sig verification overhead. \\
\addlinespace
Gharavi et al. \cite{gharavi_pqbfl_2025, gharavi_post-quantum_2024} & Secure communication protocol for FL rounds using PQC and blockchain. & Yes (Kyber KEM, potentially others, focuses on key exchange/session security via ratcheting) & Yes (Key establishment facilitation, transaction ledger, reputation mgmt) & Yes & Focus on forward secrecy/post-compromise security for iterative FL rounds using ratcheting; uses PQC KEMs more than signatures for core security. \\
\addlinespace
Ding et al. \cite{ding_post-quantum_2024} & Privacy preservation (anti-gradient leakage) in FL using PQC and SMPC. & Yes (Kyber KEM for secure channel) & Not specified / No & Yes & Focus on privacy via SMPC and PQC KEM, not blockchain verification or PQC signatures for authentication. \\
\addlinespace
Classical Secure FL \cite{bonawitz_practical_2017, geyer_differentially_2018} & Privacy (Secure Aggregation, DP), Byzantine Resilience using classical crypto. & No & Generally No (Some BFL variants exist but often use classical crypto) & Yes & Addresses different security aspects (privacy, Byzantine) but lacks quantum resistance. Serves as a baseline for non-PQC security. \\
\addlinespace
Classical BFL \cite{harris_decentralized_2019, ma_blockchain-based_2020, qammar_securing_2023} & Using blockchain for verification, incentives, auditability in FL. & No (Uses classical signatures like ECDSA) & Yes (Verification, Audit, etc.) & Yes & Demonstrates BC+FL benefits but lacks quantum resistance, providing a baseline for PQC performance comparison. \\
\bottomrule
\end{tabular}%
}
\end{table*}

By focusing on the practical integration and performance implications of using standardized PQC signatures for update verification within a blockchain-based FL system, PQS-BFL provides crucial empirical evidence supporting the transition towards quantum-resistant decentralized machine learning.

\section{Background}
\label{sec:background}

This section provides the necessary foundational concepts concerning federated learning, the principles of post-quantum cryptography with a focus on lattice-based digital signatures, and the relevant aspects of blockchain technology integration that underpin the PQS-BFL framework.

\subsection{Federated Learning (FL)}
Federated learning represents a paradigm shift in distributed machine learning, enabling collaborative model training across multiple decentralized clients (e.g., mobile devices, hospitals, organizations) without necessitating the sharing of their raw, potentially sensitive, local datasets \cite{mcmahan_communication-efficient_2017}. This approach inherently enhances data privacy compared to traditional centralized methods.

\subsubsection{Mathematical Framework}
Consider a set of $M$ clients, denoted by $C_i$ for $i \in \{1, \ldots, M\}$. Each client $C_i$ holds a local dataset $\mathcal{D}_i$, containing $n_i = |\mathcal{D}_i|$ data samples. The total dataset size across all clients is $n = \sum_{i=1}^M n_i$. The primary objective of FL is to train a global model, parameterized by a vector $w$, that minimizes a global objective function $\mathcal{L}(w)$. This global function is typically formulated as a weighted average of the local objective functions $\mathcal{L}_i(w)$, where the weights reflect the relative contribution of each client's data:
\begin{equation}
    \min_{w} \mathcal{L}(w) = \sum_{i=1}^M \frac{n_i}{n} \mathcal{L}_i(w).
\end{equation}
Here, $\mathcal{L}_i(w)$ represents the empirical loss of the model $w$ evaluated on client $C_i$'s local data $\mathcal{D}_i$. It is commonly defined as the average loss over the client's local samples:
\begin{equation}
    \mathcal{L}_i(w) = \frac{1}{n_i} \sum_{(x_j, y_j) \in \mathcal{D}_i} \ell(w; x_j, y_j),
\end{equation}
where $(x_j, y_j)$ is a data sample (features $x_j$, label $y_j$) from $\mathcal{D}_i$, and $\ell$ is a sample-wise loss function appropriate for the task (e.g., cross-entropy for classification, mean squared error for regression).

\subsubsection{Federated Averaging (FedAvg)}
Federated Averaging (FedAvg) \cite{mcmahan_communication-efficient_2017} is the canonical algorithm for achieving the FL objective. A typical communication round $t$ proceeds as follows:
\begin{enumerate}
    \item \textbf{Distribution:} A central server (or coordinating entity) distributes the current global model parameters $w_{t-1}$ to a selected subset of clients $\mathcal{S}_t \subseteq \{C_1, \ldots, C_M\}$.
    \item \textbf{Local Training:} Each selected client $C_i \in \mathcal{S}_t$ initializes its local model with $w_{t-1}$ and performs local updates using its dataset $\mathcal{D}_i$. This typically involves running multiple steps (epochs $E$) of an optimization algorithm like Stochastic Gradient Descent (SGD) on local data mini-batches to minimize $\mathcal{L}_i(w)$. Let $w_t^i$ denote the model parameters obtained by client $C_i$ after local training. Often, clients compute the model update (or delta) $\Delta w_t^i = w_t^i - w_{t-1}$.
    \item \textbf{Communication:} The participating clients transmit their updated local models $w_t^i$ (or alternatively, the updates $\Delta w_t^i$) back to the server.
    \item \textbf{Aggregation:} The server aggregates the received updates to form the new global model $w_t$. The standard aggregation rule is a weighted average based on the number of local data samples:
    \begin{equation}
        w_t = \sum_{i \in \mathcal{S}_t} \frac{n_i}{\sum_{j \in \mathcal{S}_t} n_j} w_t^i.
    \end{equation}
    (If clients send updates $\Delta w_t^i$, the server computes $w_t = w_{t-1} + \sum_{i \in \mathcal{S}_t} \frac{n_i}{\sum_{j \in \mathcal{S}_t} n_j} \Delta w_t^i$.)
\end{enumerate}
This iterative process repeats for a predetermined number of communication rounds $T$, or until a convergence criterion is met. While keeping data local enhances privacy, FL is still vulnerable to attacks targeting the communicated updates or the central server; hence robust security mechanisms, including authentication and integrity checks addressed in this work, are essential.

\subsection{Post-Quantum Cryptography (PQC)}
Post-quantum cryptography encompasses cryptographic algorithms engineered to withstand attacks mounted by both classical and powerful quantum computers. Shor's algorithm \cite{shor_polynomial-time_1999} demonstrated that quantum computers could efficiently break widely used public-key cryptosystems like RSA and ECC, which rely on the difficulty of factoring large integers and computing discrete logarithms, respectively. PQC aims to build cryptographic primitives based on mathematical problems believed to remain hard even for quantum algorithms. Digital signatures are particularly critical for ensuring authenticity and integrity in distributed systems like FL.

\subsubsection{Digital Signatures}
A digital signature scheme $\Pi$ is a tuple of three polynomial-time algorithms:
\begin{itemize}
    \item $\texttt{KeyGen}(1^\lambda) \rightarrow (pk, sk)$: On input of a security parameter $\lambda$, this algorithm generates a public key $pk$ and a corresponding secret key $sk$. The public key $pk$ can be widely distributed, while the secret key $sk$ must be kept private by the signer.
    \item $\texttt{Sign}(sk, m) \rightarrow \sigma$: Given the secret key $sk$ and a message $m$, this algorithm produces a digital signature $\sigma$.
    \item $\texttt{Verify}(pk, m, \sigma) \rightarrow \{0, 1\}$ (or $\{\text{invalid}, \text{valid}\}$): Given the public key $pk$, the message $m$, and a signature $\sigma$, this algorithm outputs 1 if $\sigma$ is a valid signature on $m$ under $pk$, and 0 otherwise.
\end{itemize}
The fundamental security requirement for a digital signature scheme is \textit{existential unforgeability under a chosen message attack (EUF-CMA)}. Informally, this means that even an adversary $\mathcal{A}$ who can adaptively request signatures on messages of their choice ($m_1, m_2, \dots$) using the legitimate signer's secret key (via a signing oracle) cannot produce a valid signature $\sigma^*$ on a \emph{new} message $m^*$ (i.e., $m^* \notin \{m_1, m_2, \dots\}$) that verifies correctly under the public key $pk$. A PQC signature scheme must maintain EUF-CMA security against adversaries equipped with quantum computational power.

\subsubsection{Lattice-Based Cryptography}
A significant family of PQC candidates is based on the presumed computational hardness of problems defined over mathematical lattices. A lattice $\mathcal{L}$ is a discrete subgroup of $\mathbb{R}^n$, typically represented as the set of all integer linear combinations of a basis $B = \{b_1, \ldots, b_k\} \subset \mathbb{R}^n$:
$$ \mathcal{L}(B) = \left\{ \sum_{i=1}^k z_i b_i \mid z_i \in \mathbb{Z} \right\}. $$
Prominent hard problems on lattices include:
\begin{itemize}
    \item \textbf{Shortest Vector Problem (SVP):} Find the shortest non-zero vector $v \in \mathcal{L}$ with respect to a given norm (e.g., Euclidean norm).
    \item \textbf{Short Integer Solution (SIS):} Given a matrix $A \in \mathbb{Z}_q^{k \times \ell}$ over a finite field $\mathbb{Z}_q$, find a non-zero integer vector $z \in \mathbb{Z}^\ell$ with small norm (e.g., $\|z\|_\infty \leq \beta$) such that $Az = 0 \pmod{q}$.
    \item \textbf{Learning With Errors (LWE):} Introduced by Regev \cite{regev_lattices_2009}, LWE is a cornerstone problem. Given access to samples of the form $(a_i, b_i = \langle a_i, s \rangle + e_i \pmod q)$, where $s \in \mathbb{Z}_q^n$ is a secret vector, $a_i \in \mathbb{Z}_q^n$ are chosen uniformly at random, and $e_i \in \mathbb{Z}_q$ are small "errors" drawn from a specific distribution (e.g., discrete Gaussian), the goal is to recover the secret $s$. Variants like Ring-LWE (RLWE) and Module-LWE (MLWE) operate over polynomial rings, offering better efficiency.
\end{itemize}
These lattice problems are widely believed to be computationally intractable for both classical and quantum computers for appropriate parameter choices, making them attractive foundations for PQC.

\subsubsection{ML-DSA (Dilithium)}
The Module-LWE Digital Signature Algorithm (ML-DSA), standardized by NIST as FIPS 204 and formerly known as Dilithium \cite{ducas_crystals-dilithium_2018, lyubashevsky_crystals-dilithium_2020}, is a prominent lattice-based signature scheme selected for standardization. Its security relies on the hardness of the MLWE and Module-SIS (MSIS) problems over polynomial rings $R_q = \mathbb{Z}_q[X]/(X^n+1)$, where $X^n+1$ is often a cyclotomic polynomial. Key characteristics include:
\begin{itemize}
    \item \textbf{Security:} Based on the well-studied hardness assumptions of MLWE and MSIS, offering strong guarantees against known classical and quantum attacks.
    \item \textbf{Efficiency:} Provides competitive performance, particularly fast signing and verification times compared to many other PQC signature candidates.
    \item \textbf{Compactness:} While signatures and keys are larger than classical counterparts like ECDSA, they are generally considered practical for a wide range of applications.
\end{itemize}
ML-DSA employs the Fiat-Shamir transformation with Aborts. This technique converts an interactive identification protocol (where a prover convinces a verifier of their identity) into a non-interactive signature scheme using a cryptographic hash function to simulate the verifier's challenges. The "abort" mechanism involves rejection sampling during signature generation: intermediate values derived from the secret key are checked, and if they fall outside a predetermined "safe" range (e.g., their norm is too large, potentially leaking key information), the signing process is aborted and restarted with fresh randomness. This ensures that the distribution of output signatures does not reveal statistical information about the secret key. ML-DSA-65 corresponds to NIST Security Level 2, aiming for security comparable to or greater than 128-bit symmetric key algorithms against quantum attackers.

\subsection{Blockchain Integration}
Blockchain technology provides a decentralized, distributed, and often immutable ledger. Transactions are grouped into blocks, cryptographically linked using hash pointers, and validated by a network consensus mechanism, ensuring consistency and tamper-resistance once blocks are finalized. In the context of FL, blockchain can enhance transparency, auditability, and trust by providing a shared, verifiable record of operations.

\subsubsection{Smart Contracts}
Smart contracts are self-executing contracts with the terms of the agreement directly written into code. They run on the blockchain network and are executed automatically and deterministically by the network nodes when predefined conditions are met (typically triggered by a transaction). In PQS-BFL, smart contracts deployed on the blockchain orchestrate the verification process:
\begin{itemize}
    \item Store immutable records, such as the registered PQC public keys ($pk_i$) associated with each participating client $C_i$.
    \item Define functions, like $\texttt{submitUpdate}(h_t^i, \sigma_t^i)$, which accept transaction inputs (e.g., the hash of a model update $h_t^i$ and its signature $\sigma_t^i$).
    \item Execute predefined logic, crucially including the $\texttt{Verify}(pk_i, h_t^i, \sigma_t^i)$ algorithm of the PQC scheme, using the submitted data and the stored public key $pk_i$.
    \item Update the blockchain state or emit events based on the verification outcome (e.g., recording verified hashes or logging verification success/failure).
\end{itemize}

\subsubsection{Gas Costs}
On many public blockchain platforms, particularly Ethereum and its derivatives, executing transactions and smart contract operations requires computational resources from the network nodes. This computational effort is quantified using a unit called "gas". Each low-level operation (opcode) within the smart contract execution has a specific gas cost. Users initiating transactions must pay a fee, typically calculated as the total gas consumed by the transaction multiplied by a gas price they specify (e.g., in Gwei). Implementing complex cryptographic operations like PQC verification directly within a smart contract can consume significant amounts of gas due to the computational complexity and potentially large data handling involved. Therefore, minimizing gas costs through efficient contract design and careful consideration of on-chain versus off-chain computation is crucial for economic viability \cite{albert_gasol_2020}.

\subsection{Integration Challenges}
Combining FL, PQC, and Blockchain presents specific challenges. The larger signature and key sizes inherent in many PQC schemes (compared to ECDSA) increase communication bandwidth requirements in FL and raise the storage and transaction data costs on the blockchain. Furthermore, the computational complexity of PQC verification algorithms, while fast in absolute terms on standard hardware, can translate into high gas costs within the constrained execution environment of smart contracts. PQS-BFL is designed to navigate these trade-offs, ensuring quantum-resistant security while maintaining practical performance.

\section{System Design and Methodology}
\label{sec:design}

This section details the architecture, protocol specification, and security analysis of the PQS-BFL framework.

\subsection{System Architecture}
PQS-BFL integrates Federated Learning (FL), Post-Quantum Cryptography (PQC), and Blockchain technology into a cohesive system. The architecture is designed modularly to separate concerns and facilitate analysis (see Figure~\ref{fig:system_architecture}).

\begin{figure}[!htbp]
    \centering
    \includegraphics[width=\linewidth]{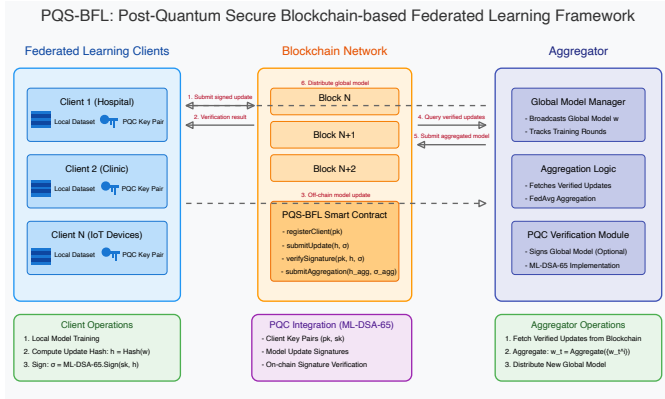}
    \caption{Overview of the PQS-BFL System Architecture. Clients perform local training, sign update hashes using PQC (ML-DSA-65), and submit them to a blockchain smart contract for verification and recording. An aggregator (which could be decentralized or a trusted server) fetches verified updates to compute the global model.}
    \label{fig:system_architecture}
\end{figure}

The core components are:
\begin{itemize}
    \item \textbf{FL Clients ($C_i$):} Entities (e.g., hospitals, devices) holding local data $\mathcal{D}_i$. They perform local model training, generate PQC key pairs $(pk_i, sk_i)$ using $\texttt{ML-DSA-65.KeyGen}$, compute hashes of their model updates $h_t^i = \texttt{Hash}(w_t^i)$, sign these hashes $\sigma_t^i = \texttt{ML-DSA-65.Sign}(sk_i, h_t^i)$, and submit $(h_t^i, \sigma_t^i)$ to the blockchain.
    \item \textbf{PQC Module:} Implements the ML-DSA-65 signature scheme ($\texttt{KeyGen}$, $\texttt{Sign}$, $\texttt{Verify}$). In our implementation, this utilizes liboqs.
    \item \textbf{Blockchain Network:} A distributed ledger (e.g., Ethereum) hosting the PQS-BFL smart contract.
    \item \textbf{Smart Contract ($\mathcal{SC}$):} Resides on the blockchain. Its functions include:
        \begin{itemize}
            \item $\texttt{registerClient}(pk_i)$: Stores the PQC public key $pk_i$ for a client.
            \item $\texttt{submitUpdate}(h_t^i, \sigma_t^i)$: Receives an update hash and signature from client $i$. It retrieves $pk_i$ and calls $\texttt{ML-DSA-65.Verify}(pk_i, h_t^i, \sigma_t^i)$. If valid, it records the hash $h_t^i$ (or emits an event).
            \item (Optional) $\texttt{submitAggregation}(h_t^{\text{agg}}, \sigma_t^{\text{agg}})$: Allows an authorized aggregator to submit the hash of the aggregated global model, signed with an aggregator key $sk_{\text{agg}}$.
        \end{itemize}
    \item \textbf{Aggregator:} An entity (potentially decentralized or a designated server) responsible for fetching verified update hashes $\{h_t^i\}$ from the blockchain (or corresponding full updates communicated off-chain), computing the aggregated global model $w_t$, and potentially distributing it for the next round.
\end{itemize}

\subsection{Protocol Specification}
The PQS-BFL protocol operates in rounds. Algorithm~\ref{alg:pqs_bfl_protocol} outlines the core steps.

\begin{algorithm}[!ht]
\caption{PQS-BFL Protocol}\label{alg:pqs_bfl_protocol}
\begin{algorithmic}[1]
\Require Global model $w_{t-1}$, Client set $\mathcal{C}$, Datasets $\{\mathcal{D}_i\}$, PQC scheme $\Pi = (\texttt{KeyGen}, \texttt{Sign}, \texttt{Verify})$, Hash function $\texttt{Hash}$, Smart Contract $\mathcal{SC}$
\Ensure Updated global model $w_t$, Blockchain record of verified updates

\Statex \textbf{Initialization Phase (Once per client):}
\For{each client $C_i \in \mathcal{C}$}
    \State $(pk_i, sk_i) \gets \Pi.\texttt{KeyGen}(1^\lambda)$
    \State Send transaction to $\mathcal{SC}.\texttt{registerClient}(pk_i)$ associated with client $i$'s blockchain address.
\EndFor

\Statex \textbf{Federated Round $t$:}
\State \textbf{Server/Aggregator:} Broadcasts $w_{t-1}$ to participating clients.
\For{each participating client $C_i$} \Comment{Client-side}
    \State $w_t^i \gets \texttt{LocalTraining}(w_{t-1}, \mathcal{D}_i)$ \Comment{Train locally}
    \State $h_t^i \gets \texttt{Hash}(w_t^i)$ \Comment{Compute hash of update/model}
    \State $\sigma_t^i \gets \Pi.\texttt{Sign}(sk_i, h_t^i)$ \Comment{Sign hash using PQC}
    \State Send transaction to $\mathcal{SC}.\texttt{submitUpdate}(h_t^i, \sigma_t^i)$
    \State (Optional: Communicate full update $w_t^i$ off-chain to aggregator, linked to $h_t^i$)
\EndFor
\State \textbf{Server/Aggregator:}
\State Wait for a sufficient number of updates to be verified on the blockchain.
\State Fetch verified hashes $\{h_t^i\}_{i \in \mathcal{V}_t}$ from $\mathcal{SC}$, where $\mathcal{V}_t$ is the set of verified clients in round $t$.
\State Obtain corresponding full updates $\{w_t^i\}_{i \in \mathcal{V}_t}$ (e.g., via off-chain channel).
\State $w_t \gets \texttt{Aggregate}(\{w_t^i\}_{i \in \mathcal{V}_t})$ \Comment{e.g., FedAvg}
\State (Optional: $h_t^{\text{agg}} \gets \texttt{Hash}(w_t)$; $\sigma_t^{\text{agg}} \gets \Pi.\texttt{Sign}(sk_{\text{agg}}, h_t^{\text{agg}})$)
\State (Optional: Send transaction to $\mathcal{SC}.\texttt{submitAggregation}(h_t^{\text{agg}}, \sigma_t^{\text{agg}})$)
\State \Return $w_t$
\end{algorithmic}
\end{algorithm}

\textbf{Key Design Choices:}
\begin{itemize}
    \item \textbf{Hashing Updates:} Clients sign the \textit{hash} of their model update rather than the full update. This significantly reduces the data that needs to be signed and potentially stored/processed on the blockchain, making it more efficient, especially for large models. A secure hash function (e.g., SHA3-256) is assumed.
    \item \textbf{On-Chain vs. Off-Chain Data:} Only the update hash $h_t^i$ and signature $\sigma_t^i$ are submitted on-chain for verification. The potentially large model update $w_t^i$ can be communicated off-chain directly to the aggregator, linked via the verified hash. This minimizes blockchain storage and gas costs.
    \item \textbf{Aggregator Role:} The aggregator role can be centralized or decentralized. In our experiments, we use a central aggregator for simplicity, but the blockchain verification mechanism supports decentralized aggregation designs.
\end{itemize}

Figure~\ref{fig:sequence_diagram} illustrates the complete sequence of operations and message exchanges in the PQS-BFL protocol, highlighting the interactions between clients, smart contracts, blockchain network, and the aggregator.

\begin{figure}[!ht]
    \centering
    \includegraphics[width=\linewidth]{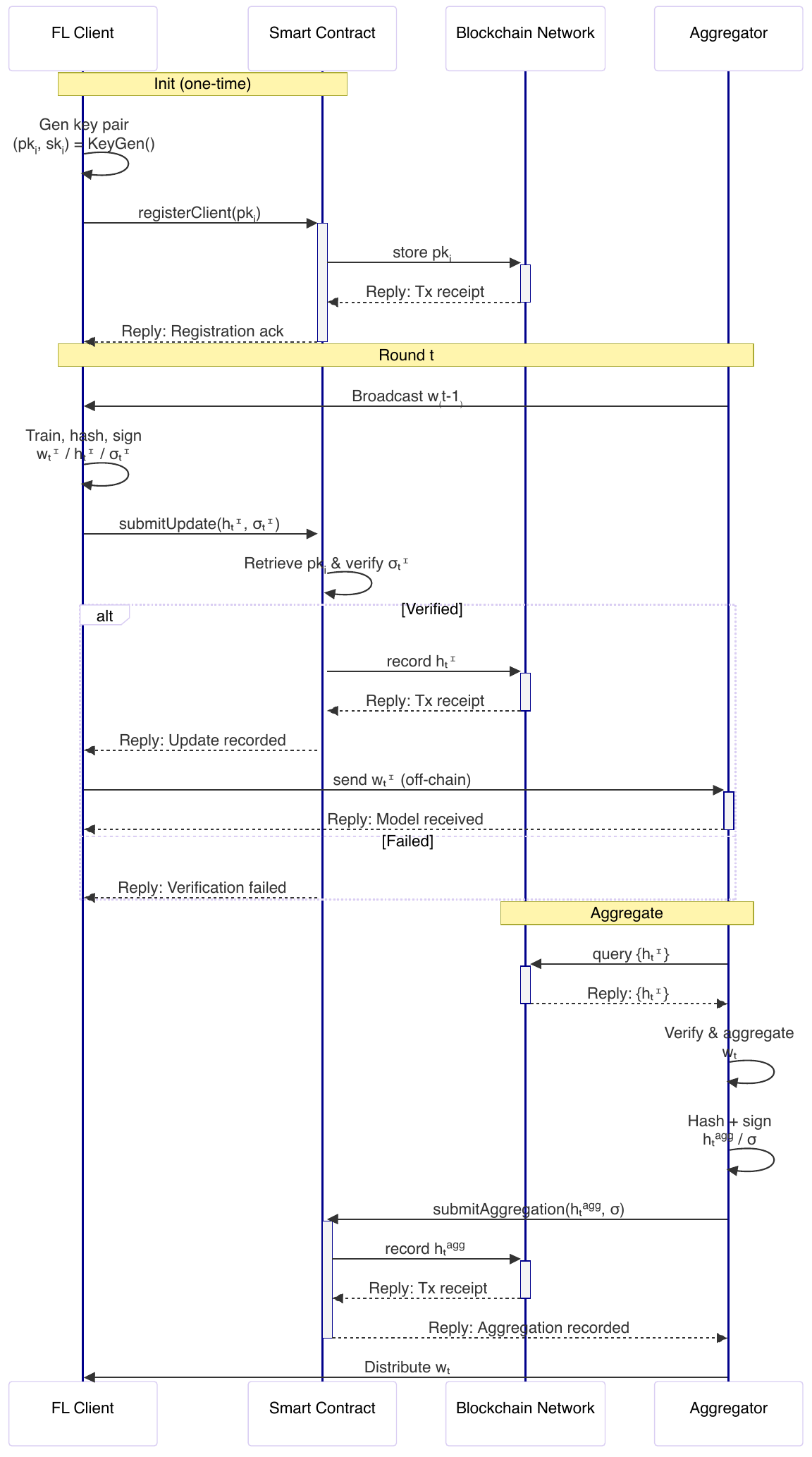}
    \caption{Sequence diagram of the PQS-BFL protocol, showing interactions between FL clients, smart contract, blockchain network, and aggregator. The diagram illustrates the end-to-end process from initialization to model distribution, highlighting the use of ML-DSA-65 signatures for update verification.}
    \label{fig:sequence_diagram}
\end{figure}

\subsection{Security Analysis}

\subsubsection{Threat Model}
We consider a powerful adversary $\mathcal{A}$ equipped with quantum computational capabilities (modeled as a Quantum Polynomial-Time (QPT) adversary). The adversary's goal is to compromise the integrity or authenticity of the federated learning process mediated by the blockchain. We assume $\mathcal{A}$ can:
\begin{itemize}
    \item Passively observe all public communication, including all transactions broadcast on the blockchain network and any off-chain communication it can intercept.
    \item Actively inject, modify, or drop messages in the network, subject to the guarantees provided by the underlying network and blockchain protocols.
    \item Adaptively corrupt a subset of FL clients, thereby obtaining their PQC secret keys $sk_i$ and potentially influencing their computations.
    \item Interact with the PQS-BFL smart contract $\mathcal{SC}$ by sending transactions (e.g., attempting to register keys or submit updates).
    \item Attempt to forge PQC signatures on model update hashes ($h_t^i$) for uncorrupted clients.
    \item Attempt to manipulate the aggregation process, for example, by trying to inject malicious updates or influence which updates are considered valid.
\end{itemize}
We detail the specific security assumptions regarding the cryptographic primitives and the blockchain platform below.

\subsubsection{Security Goals}
The primary security objectives of PQS-BFL against the described adversary are:
\begin{enumerate}
    \item \textbf{Update Authenticity:} Ensure that only the client $C_i$ possessing the secret key $sk_i$ corresponding to the registered public key $pk_i$ can produce a valid signature $\sigma_t^i$ on an update hash $h_t^i$ that will be accepted as valid for $C_i$ by the smart contract $\mathcal{SC}$.
    \item \textbf{Update Integrity:} Guarantee that any modification to a model update $w_t^i$ after it has been hashed to $h_t^i$ will render the original signature $\sigma_t^i$ invalid for the modified content. Furthermore, ensure that any tampering with the hash $h_t^i$ or the signature $\sigma_t^i$ during transmission or prior to blockchain finality prevents successful verification by $\mathcal{SC}$.
    \item \textbf{Unforgeability:} Prevent the adversary $\mathcal{A}$ from generating a valid signature $\sigma^*$ for any \emph{new} update hash $h^*$ (i.e., one not previously signed by the legitimate client) under the public key $pk_i$ of an uncorrupted client $C_i$.
\end{enumerate}

\subsubsection{Security Arguments}
The security guarantees of PQS-BFL are predicated on the following standard cryptographic and system assumptions:

\begin{assumption}[PQC Signature Security]\label{assump:pqc}
The chosen Post-Quantum Cryptography digital signature scheme, $\Pi$ (specifically ML-DSA-65 in our implementation), provides existential unforgeability under a chosen message attack (EUF-CMA) against quantum polynomial-time (QPT) adversaries. ML-DSA-65 targets NIST PQC Security Level 2 \cite{alagic_status_2022}, implying resistance comparable to or exceeding AES-128 against quantum attacks. This security relies on the conjectured quantum hardness of the Module Learning With Errors (MLWE) and Module Short Integer Solution (MSIS) problems over specific polynomial rings \cite{ducas_crystals-dilithium_2018, lyubashevsky_crystals-dilithium_2020}.
\end{assumption}

\begin{assumption}[Cryptographic Hash Function Security]\label{assump:hash}
The cryptographic hash function $\texttt{Hash}$ (SHA3-256 in our implementation) used to compute $h_t^i = \texttt{Hash}(w_t^i)$ is collision-resistant and second-preimage resistant against QPT adversaries. Collision resistance means it is computationally infeasible to find two distinct inputs $w \neq w'$ such that $\texttt{Hash}(w) = \texttt{Hash}(w')$. Second-preimage resistance means that given an input $w$, it is computationally infeasible to find a different input $w' \neq w$ such that $\texttt{Hash}(w) = \texttt{Hash}(w')$.
\end{assumption}

\begin{assumption}[Blockchain Platform Security]\label{assump:blockchain}
The underlying blockchain network provides essential properties for recording and verifying FL updates:
\begin{enumerate}
    \item \textbf{Consensus and Finality:} The network nodes correctly execute a consensus protocol (e.g., Nakamoto consensus in Proof-of-Work, or a BFT variant in Proof-of-Stake) to agree on the order and validity of transactions. Once a transaction containing a verified update $(h_t^i, \sigma_t^i)$ is included in a block that achieves sufficient finality (e.g., probabilistic finality after $k$ confirmations in PoW, or deterministic finality in some PoS systems), it is considered immutable.
    \item \textbf{State Consistency:} Honest nodes maintain a consistent view of the blockchain state, including registered public keys and the record of verified update hashes.
    \item \textbf{Liveness (Partial):} Transactions submitted by honest clients are eventually included in finalized blocks, assuming appropriate transaction fees are paid.
    \item \textbf{Resistance:} The consensus mechanism is resistant to attacks by the QPT adversary $\mathcal{A}$ aiming to reverse finalized blocks or cause persistent forks (e.g., resistant to quantum-enhanced 51\% attacks up to a certain threshold).
\end{enumerate}
\textit{Crucially, we note that this assumption may be strong if the blockchain platform itself relies on classical cryptography (e.g., ECDSA for transaction signing or consensus mechanisms). In such cases, while PQS-BFL secures the FL updates with PQC, the underlying ledger's security could remain vulnerable to quantum attacks, representing a potential bottleneck in the overall system's quantum resistance. A fully quantum-resistant system would require the blockchain platform itself to employ PQC.}
\end{assumption}

Based on these assumptions, we argue the security of PQS-BFL through the following theorems:

\begin{theorem}[Update Authenticity and Integrity]
Under Assumptions \ref{assump:pqc}, \ref{assump:hash}, and \ref{assump:blockchain}, the PQS-BFL protocol ensures update authenticity and integrity. Specifically, if the smart contract $\mathcal{SC}$ records an update hash $h_t^i$ as valid for client $C_i$, then with overwhelming probability: (a) the corresponding signature $\sigma_t^i$ was generated using the secret key $sk_i$ associated with the registered public key $pk_i$, and (b) $h_t^i$ is the correctly computed hash of the model update $w_t^i$ produced by the holder of $sk_i$.
\end{theorem}
\begin{proof}[Proof Sketch]
\textbf{Authenticity:} The smart contract $\mathcal{SC}$ validates a submitted pair $(h_t^i, \sigma_t^i)$ by executing $\Pi.\texttt{Verify}(pk_i, h_t^i, \sigma_t^i)$ using the publicly registered key $pk_i$. By Assumption \ref{assump:pqc} (EUF-CMA security of $\Pi$), it is computationally infeasible for any QPT adversary $\mathcal{A}$ not possessing $sk_i$ to create a valid signature $\sigma_t^i$ on any hash $h_t^i$ under $pk_i$, except with negligible probability. Thus, a successful verification implies the signature originated from the holder of $sk_i$, establishing authenticity.

\textbf{Integrity:} Integrity stems from the binding properties of the hash function and the signature, combined with blockchain immutability.
1.  \textit{Binding Update to Hash:} Client $C_i$ computes $h_t^i = \texttt{Hash}(w_t^i)$. By Assumption \ref{assump:hash} (collision and second-preimage resistance), it is computationally infeasible for an adversary to find a different update $w' \neq w_t^i$ such that $\texttt{Hash}(w') = h_t^i$. Thus, $h_t^i$ uniquely represents $w_t^i$ for practical purposes. Any modification to $w_t^i$ before hashing would result in a different hash $h'$.
2.  \textit{Binding Hash to Client:} As established by authenticity, a valid signature $\sigma_t^i$ cryptographically binds the specific hash $h_t^i$ to the public key $pk_i$ (and thus the client $C_i$). Any modification to $h_t^i$ or $\sigma_t^i$ after signing would cause the verification $\Pi.\texttt{Verify}(pk_i, h_t^i, \sigma_t^i)$ to fail.
3.  \textit{Immutability:} By Assumption \ref{assump:blockchain}(a), once the transaction containing the verified $(h_t^i, \sigma_t^i)$ pair (or a record derived from it) achieves finality on the blockchain, it cannot be altered or removed by the adversary.

Therefore, a verified record on the blockchain guarantees that the hash $h_t^i$ corresponds authentically to the update $w_t^i$ signed by the legitimate client $C_i$, and that neither the update content (relative to its hash), the hash, nor the signature were successfully tampered with.
\end{proof}

\begin{theorem}[Resistance to Forgery]
Under Assumption \ref{assump:pqc}, the PQS-BFL protocol resists signature forgery. A QPT adversary $\mathcal{A}$, even after observing multiple valid update hashes and signatures $\{ (h_t^j, \sigma_t^j) \}_{j,t}$ and potentially obtaining signatures for chosen messages via corrupted clients, cannot produce a valid signature $\sigma^*$ for a \emph{new} update hash $h^*$ (i.e., $h^*$ was not previously signed by the uncorrupted client $C_i$) under the public key $pk_i$ of an uncorrupted client $C_i$, except with negligible probability.
\end{theorem}
\begin{proof}
This property is a direct consequence of the definition of Existential Unforgeability under a Chosen Message Attack (EUF-CMA) for the signature scheme $\Pi$. Assumption \ref{assump:pqc} explicitly states that $\Pi$ (ML-DSA-65) satisfies EUF-CMA security against QPT adversaries. The theorem's claim is precisely that the adversary cannot succeed in the EUF-CMA game against $\Pi$ using the PQS-BFL protocol interactions. Therefore, the probability of successful forgery by any QPT adversary $\mathcal{A}$ is negligible in the security parameter $\lambda$ used for generating the keys for $\Pi$.
\end{proof}

In summary, the security of PQS-BFL hinges on the quantum resistance of the ML-DSA signature scheme, the robustness of the SHA3-256 hash function, and the integrity guarantees provided by the underlying blockchain platform. While PQS-BFL significantly enhances the security of the FL update verification process against quantum threats, achieving end-to-end quantum resistance requires ensuring that the blockchain layer itself (Assumption \ref{assump:blockchain}) is also secured using PQC.

\section{Experimental Setup}
\label{sec:setup}

This section details the configuration and methodology employed to rigorously evaluate the performance and feasibility of the proposed PQS-BFL framework.

\subsection{Experimental Framework}
All experiments were conducted utilizing the PQS-BFL framework. PQS-BFL is an internally developed, Python-based platform designed specifically for the systematic benchmarking of secure, quantum-resistant federated learning approaches. It integrates several key libraries: PyTorch serves as the foundation for machine learning model implementation and training; the liboqs library (\url{https://github.com/open-quantum-safe/liboqs}), accessed via its Python bindings, provides the implementation for the post-quantum cryptographic scheme ML-DSA-65; the standard Python \texttt{cryptography} library is used for classical ECDSA operations (employing the SECP256k1 curve); and Web3.py facilitates interaction with the blockchain backend for transaction submission and data retrieval.

\subsection{Hardware and Software}
The computational experiments were performed on a desktop machine equipped with an Intel Core i5-11500 processor and \SI{32}{GB} of RAM, running the Ubuntu 24.04 LTS operating system. To ensure reproducible and controlled blockchain interactions without the variability of public testnets, the blockchain environment was simulated locally using Ganache (\url{https://trufflesuite.com/ganache/}). Ganache provides a personal Ethereum blockchain instance, allowing for consistent measurement of transaction times and gas costs in a controlled setting. The entire experimental framework was implemented using Python version 3.8.

\subsection{Experimental Configurations}
We designed the experiments to facilitate performance comparisons across several critical dimensions: the dataset used for training, the cryptographic signature scheme employed, the number of participating clients, and the integration of blockchain for verification. This resulted in a series of distinct configurations, summarized in Table~\ref{tab:crypto_comp} and Table~\ref{tab:overall_comp} (presented in Section~\ref{sec:results}). A consistent naming convention, \texttt{dataset-crypto-clients-blockchain}, is used to identify each experimental run. For example, \texttt{mnist-PQC-10c-BC} denotes an experiment using the MNIST dataset, ML-DSA-65 (PQC) signatures, with 10 participating clients, and utilizing on-chain blockchain verification (\texttt{BC}). Conversely, \texttt{NoBC} indicates a configuration where blockchain interaction is simulated with a fixed delay for comparison purposes. The cryptographic options include \texttt{PQC} (ML-DSA-65), \texttt{ECDSA} (SECP256k1), and \texttt{NONE} (a baseline using only SHA-256 hashing without digital signatures). Client counts tested were 3 (\texttt{3c}), 10 (\texttt{10c}), and 30 (\texttt{30c}).

\subsection{Datasets and Models}
To evaluate the framework's performance across different data types and model complexities, we utilized three distinct datasets:
\begin{itemize}
    \item \textbf{MNIST}: A widely used benchmark dataset consisting of 28$\times$28 pixel grayscale images of handwritten digits (10 classes). For this dataset, we employed a LeNet-style Convolutional Neural Network (CNN) with approximately \num{21800} trainable parameters. Data was distributed non-identically and non-independently (non-IID) among clients using a Dirichlet distribution with parameter $\alpha=0.5$.
    \item \textbf{SVHN}: The Street View House Numbers dataset, composed of 32$\times$32 pixel color images of house numbers (10 classes). A custom CNN architecture, larger than the one used for MNIST (approximately \num{237600} parameters), was trained on this dataset. A non-IID data split was also employed for SVHN, simulating heterogeneity across clients.
    \item \textbf{HAR}: The UCI Human Activity Recognition Using Smartphones dataset. This dataset involves classifying human activities based on sensor data, represented by 561 features (6 classes). A Multi-Layer Perceptron (MLP) model with approximately \num{42900} parameters was used. Data was distributed in a non-IID fashion among clients.
\end{itemize}

\subsection{Federated Learning Parameters}
The following parameters were consistently used for the Federated Learning process across all experiments. The core algorithm employed was Federated Averaging (FedAvg). Training proceeded for a total of 50 communication rounds. In each round, all clients designated for that specific experiment configuration participated in the training (i.e., no client sampling). Each participating client performed local training for $E=5$ epochs on its local dataset before submitting its update. The Adam optimizer was used for local model updates with a learning rate set to 0.001. A batch size of 64 was used during local training.

\subsection{Cryptographic and Blockchain Parameters}
The specific cryptographic configurations were as follows: Post-quantum cryptography was represented by ML-DSA-65, implemented using the liboqs library. The classical baseline was ECDSA using the SECP256k1 curve, with SHA-256 hashing for message digesting prior to signing, implemented via the Python \texttt{cryptography} library. For hashing model updates before signing (in PQC, ECDSA, and NONE configurations), the SHA3-256 algorithm was used. The blockchain backend was a local Ganache instance configured with default gas pricing and block gas limits. Smart contracts were written in Solidity version 0.8.x. Since precompiled contracts for ML-DSA-65 verification are not standard, gas costs for PQC verification on the blockchain were estimated based on operational complexity relative to known operations, providing a reasonable approximation for comparison. The 'NoBC' configuration simulated the absence of blockchain latency by introducing a fixed delay of \SI{0.05}{s} in place of the on-chain transaction submission and confirmation time.

\subsection{Evaluation Metrics}
We measured a comprehensive set of metrics to evaluate the different facets of the PQS-BFL framework's performance. Model performance was assessed using the test-set accuracy achieved by the global model after each aggregation round. Cryptographic performance was measured by timing the key generation (\texttt{KeyGen}), signing (\texttt{Sign}), and verification (\texttt{Verify}) operations using Python's \texttt{time.perf\_counter()} for high precision, with values averaged over multiple runs and rounds. The sizes of public keys, private keys, and generated signatures were recorded in bytes. Blockchain overhead was quantified by measuring the average gas used per transaction, obtained from the receipts generated by the Ganache instance. The relative cryptographic overhead was calculated as the ratio of the combined Sign and Verify times to the total transaction time (or the fixed \SI{0.05}{s} delay in NoBC scenarios). Additionally, we calculated blockchain efficiency metrics, including the number of FL rounds accommodated per blockchain transaction (relevant if batching were used, though 1:1 in our setup), the average total gas consumed per FL round (summing gas over all client updates in a round), and the accuracy gain achieved per unit of gas consumed, offering insight into the cost-effectiveness of learning progress. All time-based and cost-based metrics were averaged over the 50 FL rounds conducted for each experiment.

\section{Experimental Results and Discussion}
\label{sec:results}

This section presents and analyzes the empirical results obtained from executing the experimental configurations detailed in Section~\ref{sec:setup}. We examine cryptographic performance, its impact on federated learning, blockchain overheads, scalability, and cross-dataset behavior.

\subsection{Cryptographic Performance Comparison}
We first evaluate the fundamental performance characteristics of the cryptographic schemes themselves. Table~\ref{tab:crypto_comp} provides average timings for key generation, signing, and verification, along with the corresponding key and signature sizes for ML-DSA-65 (PQC), ECDSA (SECP256k1), and the baseline hashing approach (NONE). Figure~\ref{fig:crypto_perf} visualizes these comparisons for the MNIST dataset configuration.

\begin{table*}[ht]
\caption{Cryptographic Performance Comparison (Average Values)}
\label{tab:crypto_comp}
\centering
\footnotesize
\sisetup{round-mode=places,round-precision=3}
\begin{tabular}{lll
  S[round-precision=3]
  S[round-precision=3]
  S[round-precision=3]
  S[round-precision=0]
  S[round-precision=0]
  S[round-precision=0]}
\toprule
Configuration & Dataset & Crypto & {KeyGen (ms)} & {Sign (ms)} & {Verify (ms)} & {Sig Size (B)} & {PubKey (B)} & {PrivKey (B)} \\
\midrule
har-PQC-3c-BC     & HAR   & PQC   & 0.147 & 0.659 & 0.538 & 3309 & 1952 & 4032 \\
mnist-ECDSA-3c-BC & MNIST & ECDSA & 0.917 & 0.163 & 0.145 &   71 &  178 &  241 \\
mnist-NONE-3c-BC  & MNIST & NONE  & 0.000 & 0.010 & 0.004 &   32 &   26 &   27 \\
mnist-PQC-10c-BC  & MNIST & PQC   & 0.190 & 0.656 & 0.546 & 3309 & 1952 & 4032 \\
mnist-PQC-30c-BC  & MNIST & PQC   & 0.127 & 0.654 & 0.525 & 3309 & 1952 & 4032 \\
mnist-PQC-3c-BC   & MNIST & PQC   & 0.137 & 0.640 & 0.527 & 3309 & 1952 & 4032 \\
mnist-PQC-3c-NoBC & MNIST & PQC   & 0.150 & 0.609 & 0.524 & 3309 & 1952 & 4032 \\
svhn-PQC-3c-BC    & SVHN  & PQC   & 0.114 & 0.636 & 0.522 & 3309 & 1952 & 4032 \\
\bottomrule
\end{tabular}
\end{table*}

\begin{figure}[!htbp]
\centering
\includegraphics[width=\linewidth]{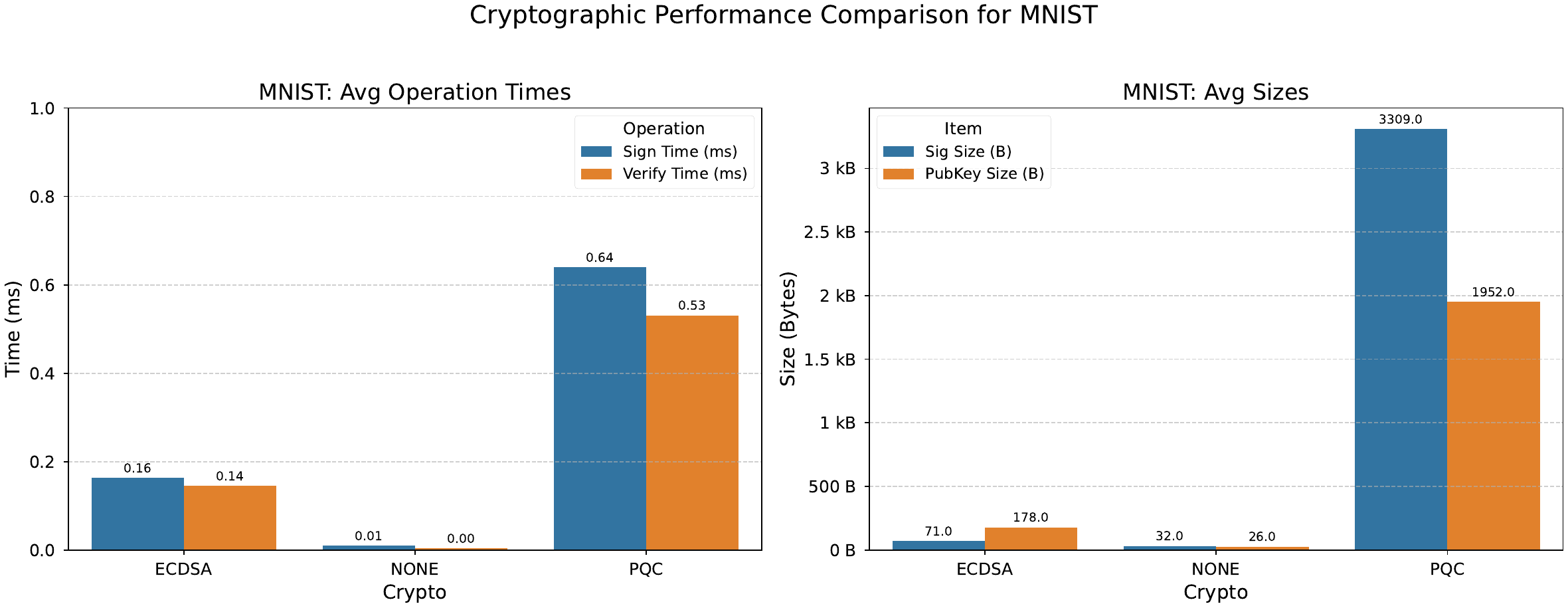}
\caption{Cryptographic performance comparison for the MNIST dataset configuration, illustrating the average operation times (left, logarithmic scale) and key/signature sizes (right) for ML-DSA-65 (PQC), ECDSA, and NONE (hashing only).}
\label{fig:crypto_perf}
\end{figure}

The results indicate that the core ML-DSA-65 operations are highly efficient, consistently achieving sub-millisecond timings. On average, PQC key generation takes approximately \SI{0.15}{ms}, signing requires around \SI{0.65}{ms}, and verification completes in about \SI{0.53}{ms}. Compared to classical ECDSA, PQC signing is roughly 4 times slower (\SI{0.65}{ms} vs. \SI{0.16}{ms}), and PQC verification is about 3.6 times slower (\SI{0.53}{ms} vs. \SI{0.15}{ms}). The NONE configuration, involving only hashing, is understandably the fastest. A significant difference lies in the data sizes: ML-DSA-65 signatures (\SI{3309}{Bytes}) and public/private keys (\SI{1952}{B}/\SI{4032}{B}) are substantially larger than their ECDSA counterparts (\SI{71}{B} signature, \SI{178}{B}/\SI{241}{B} keys), which has implications for communication and storage overhead.

\subsection{Impact on Federated Learning Performance}
We next assessed how the integration of different cryptographic schemes affects the overall federated learning process, particularly model accuracy and training time per round. Table~\ref{tab:overall_comp} presents these key FL performance metrics, while Figure~\ref{fig:training_convergence} shows the convergence behavior (accuracy and loss) for the MNIST dataset across the different cryptographic configurations.

\begin{table*}[!htbp]
\caption{Overall Performance Metrics}
\footnotesize{Average values across 50 rounds. 'Avg. Tx Time / Client (s)' and 'Avg. Gas Used / Client Tx' represent the average blockchain transaction confirmation time and gas consumption for a single client's update submission in each round, respectively. 'Round Time (s)' is the total wall-clock time for one complete federated learning round.}
\label{tab:overall_comp}
\centering
\sisetup{round-mode=places,round-precision=2,table-format=1.3}
\begin{tabular}{l
  S[round-precision=3]
  S
  S[round-precision=3]
  S[round-precision=0]
  S[round-precision=2] 
  S[scientific-notation=true, round-precision=4, table-format=1.4e1]
  S[round-precision=6]}
\toprule
\footnotesize{Configuration} & {\footnotesize{Accuracy (\%)}} & {\footnotesize{Round Time (s)}} & {\footnotesize{Sign (ms)}} & {\footnotesize{Sig Size (B)}} & {\footnotesize{Avg. Tx Time (s)}} & {\footnotesize{Avg. Gas Used}} & {\footnotesize{Overhead Ratio}} \\
\midrule
har-PQC-3c-BC     & 94.905 &  3.05 & 0.659 & 3309 & 0.33 & 1.7241e+06 & 0.001939 \\
mnist-ECDSA-3c-BC & 99.069 &  8.14 & 0.163 &   71 & 0.12 & 1.8890e+05 & 0.000380 \\
mnist-NONE-3c-BC  & 99.319 &  8.12 & 0.010 &   32 & 0.11 & 1.7365e+05 & 0.000031 \\
mnist-PQC-3c-NoBC & 99.409 &  6.88 & 0.609 & 3309 & {N/A} & {N/A}      & 0.023099 \\
mnist-PQC-3c-BC   & 99.189 &  9.01 & 0.640 & 3309 & 0.32 & 1.7241e+06 & 0.000811 \\
mnist-PQC-10c-BC  & 98.968 & 10.30 & 0.656 & 3309 & 0.05 & 5.1723e+05 & 0.000347 \\
mnist-PQC-30c-BC  & 98.818 & 14.85 & 0.654 & 3309 & 0.03 & 1.7241e+05 & 0.000167 \\
svhn-PQC-3c-BC    & 93.819 & 16.30 & 0.636 & 3309 & 0.37 & 1.7241e+06 & 0.001008 \\
\bottomrule
\end{tabular}
\end{table*}

\begin{figure}[!htbp]
\centering
\includegraphics[width=\linewidth]{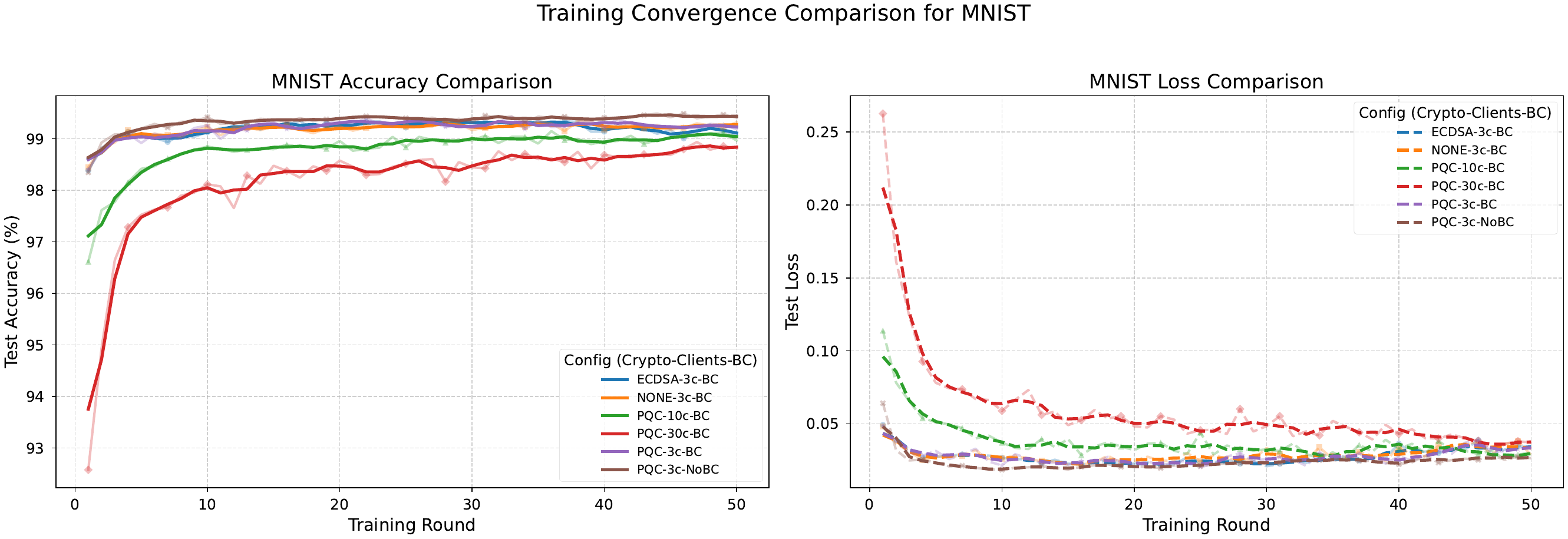}
\caption{Training convergence comparison for the MNIST dataset (3 clients, blockchain enabled). Left: Smoothed test accuracy over FL rounds. Right: Smoothed training loss over FL rounds. Curves show comparable learning performance across PQC, ECDSA, and NONE configurations.}
\label{fig:training_convergence}
\end{figure}

The results demonstrate that the choice of signature scheme has minimal impact on the final model accuracy. For instance, on MNIST with 3 clients, PQC achieved 99.19\% accuracy, ECDSA reached 99.07\%, and the NONE baseline attained 99.32\%, all highly comparable. The overall time required per federated learning round is primarily dictated by the complexity of the local training task (dataset size, model size, local epochs) rather than the cryptographic operations. For example, the simpler HAR task averaged \SI{3.05}{s} per round, MNIST averaged around \SI{8}{}--\SI{9}{s}, and the more complex SVHN task took \SI{16.30}{s} per round (for the 3-client PQC configuration). While PQC operations are slower than ECDSA, their contribution to the total round time is small, leading to only minor increases in overall round duration compared to ECDSA or NONE baselines.

\subsection{Blockchain Performance and Gas Usage}
Integrating verification into the blockchain introduces costs in terms of gas consumption and transaction confirmation time. Our experiments using Ganache provide estimates for these overheads. The average gas consumed for verifying an ML-DSA-65 signature within the smart contract was estimated at approximately \num{1.72e6} units. This is substantially higher, roughly 9 times more, than the gas required for ECDSA verification (around \num{1.89e5} units) or the baseline hash check (around \num{1.74e5} units), primarily due to the larger signature size and more complex verification arithmetic potentially involved in PQC. Consequently, the average transaction confirmation time for PQC updates (\SIrange{0.32}{0.37}{s} in our local setup for 3-client scenarios) was observed to be 2 to 3 times longer than for ECDSA (\SI{0.12}{s}) or NONE (\SI{0.11}{s}).

\subsection{Crypto Overhead Analysis}
A crucial aspect is understanding the overhead of the cryptographic operations (signing and verifying) relative to the time taken for the entire blockchain transaction. Figure~\ref{fig:crypto_overhead} plots this overhead ratio.

\begin{figure}[!htbp]
\centering
\includegraphics[width=0.9\linewidth]{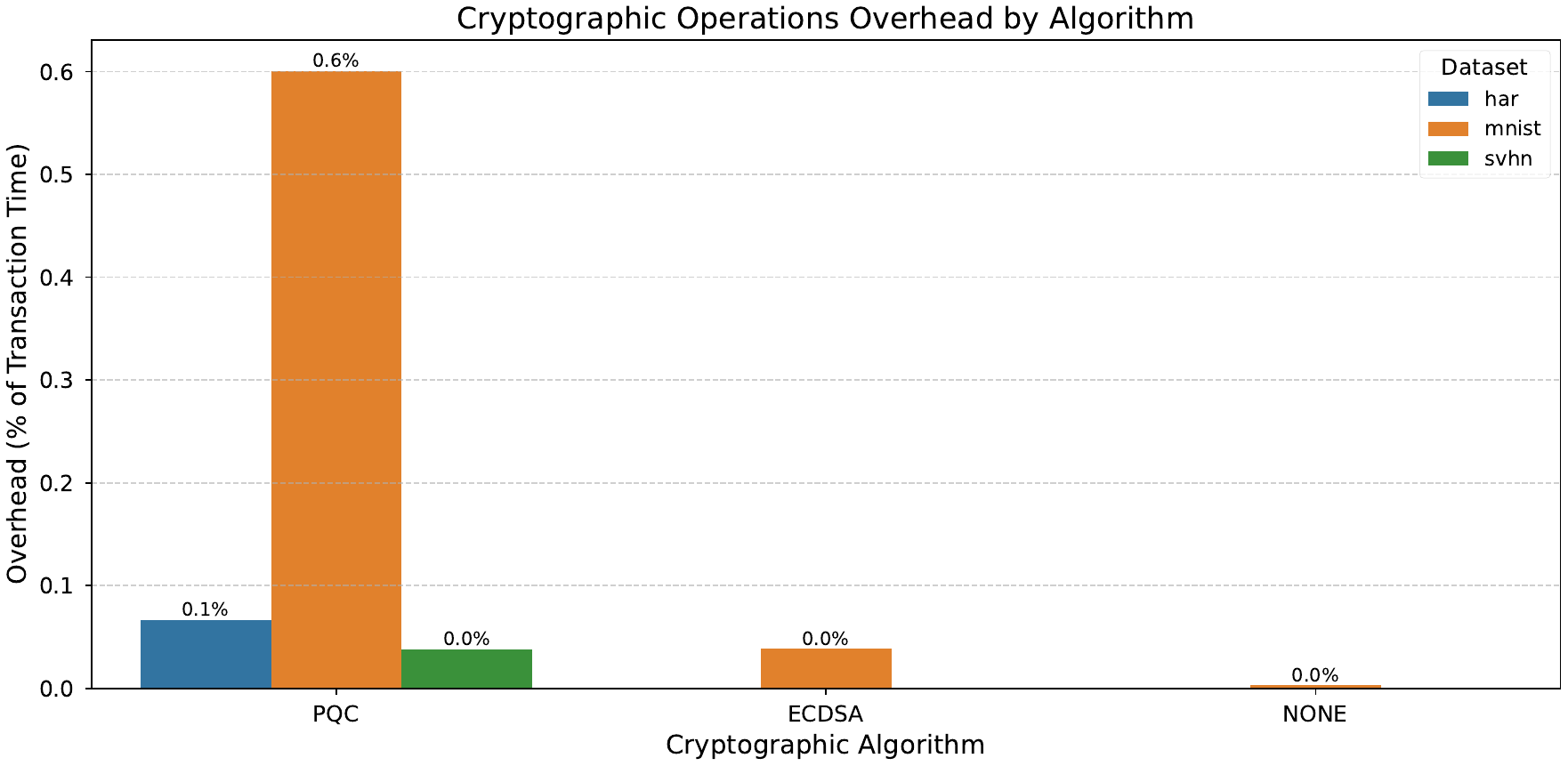}
\caption{Cryptographic overhead analysis. The bars show the combined time for signing and verifying as a percentage of the total blockchain transaction time (or fixed delay for NoBC). PQC overhead is negligible when blockchain latency is considered.}
\label{fig:crypto_overhead}
\end{figure}

The analysis reveals a key finding: when PQC signatures are verified on the blockchain (PQC+BC configurations), the time spent on the cryptographic operations themselves constitutes a very small fraction of the total transaction time. The overhead ratio consistently falls between 0.000167 (0.0167\%) and 0.001939 (0.19\%), rendering it negligible. This indicates that the blockchain network latency and consensus process dominate the transaction time, not the PQC computation. The ECDSA+BC overhead ratio is similarly low (0.000380 or 0.038\%). In contrast, the PQC+NoBC configuration, which replaces blockchain time with a fixed \SI{0.05}{s} delay, shows a much higher cryptographic overhead ratio (0.023099 or 2.3\%), highlighting that PQC performance is noticeable relative to very short delays but becomes insignificant compared to typical blockchain latencies.

\subsection{Blockchain Efficiency Metrics}
To further understand the cost-effectiveness, we examined metrics relating FL progress to blockchain resource usage. The average gas consumed per complete FL round (aggregating costs from all clients in that round) for PQC configurations (e.g., $\approx$ \SI{5.7e6}{} gas for MNIST-3c-BC) is significantly higher than for ECDSA or NONE configurations (e.g., $\approx$ \SI{9e5}{} gas for MNIST-3c-ECDSA/NONE). However, evaluating efficiency purely based on gas cost per round can be misleading. Metrics like accuracy gain per unit of gas consumed provide a more nuanced view, showing variations based on the dataset, client count, and stage of training. While PQC incurs higher gas costs per update verification, its impact on the overall cost-effectiveness of achieving a target accuracy level requires considering the entire training process and potential real-world transaction fees.

\subsection{Scalability Analysis}
We investigated how the system scales as the number of participating clients increases, focusing on the MNIST dataset with PQC signatures and blockchain verification. Figure~\ref{fig:scalability_round} shows the trend in average round time, and Figure~\ref{fig:scalability_tx} shows the average per-client transaction time.

\begin{figure}[!htbp]
\centering
\includegraphics[width=0.7\linewidth]{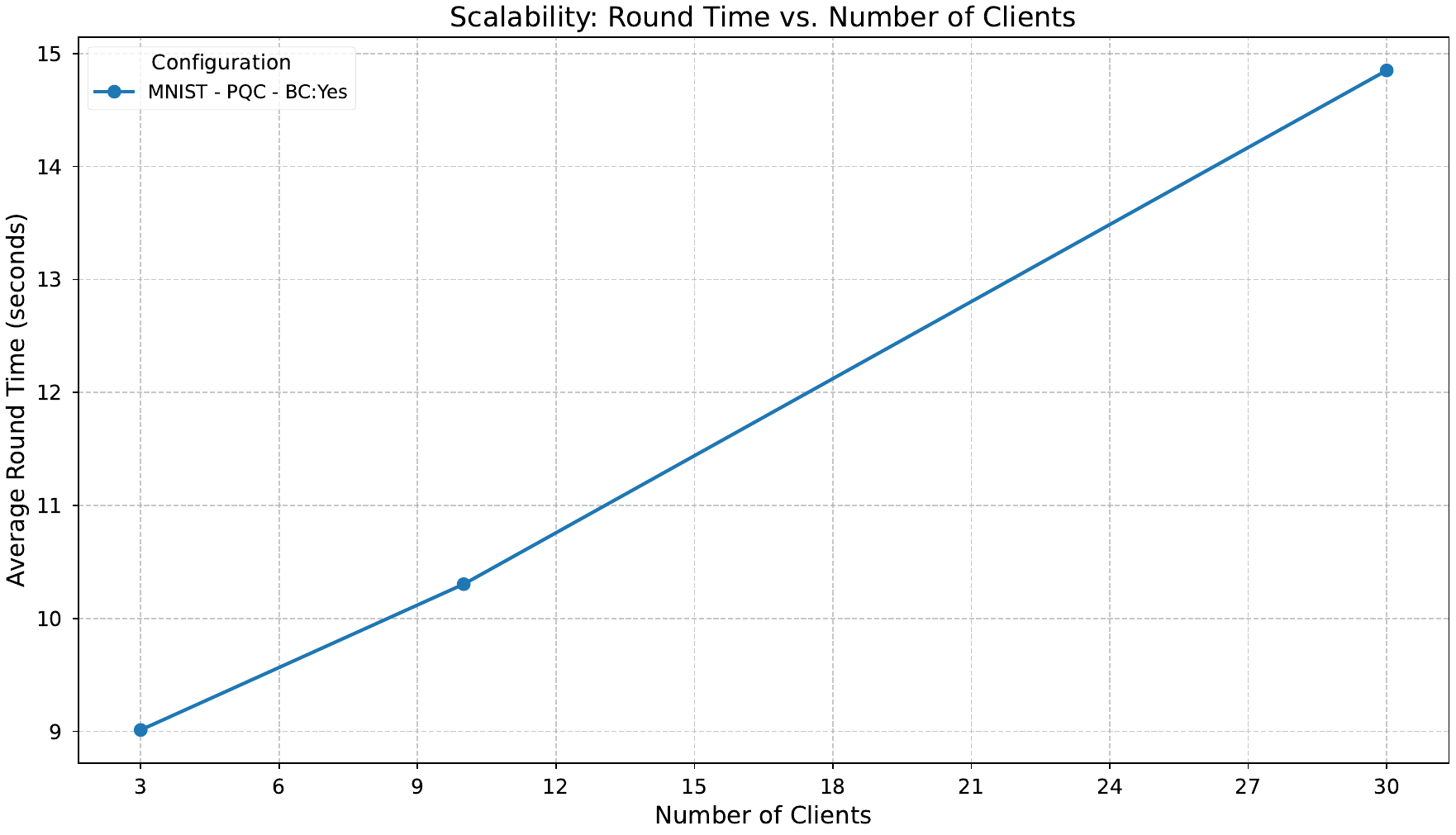}
\caption{Scalability of average round time with increasing number of clients for the MNIST PQC+BC configuration. The growth is sub-linear, indicating good scalability.}
\label{fig:scalability_round}
\end{figure}

\begin{figure}[!htbp]
\centering
\includegraphics[width=0.7\linewidth]{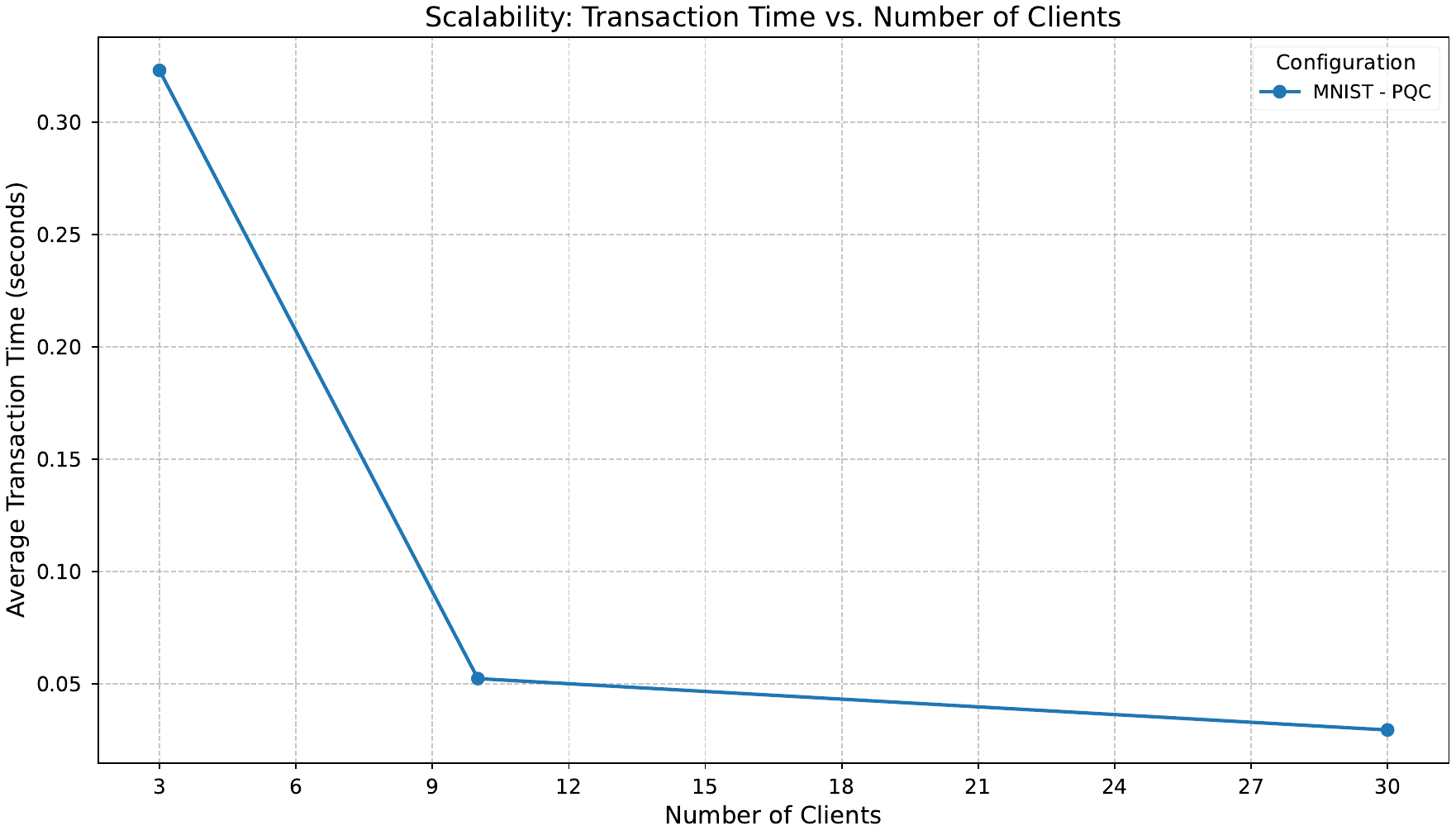}
\caption{Scalability of average per-client blockchain transaction time with increasing number of clients for the MNIST PQC+BC configuration. Transaction time remains relatively stable.}
\label{fig:scalability_tx}
\end{figure}

The results demonstrate favorable scalability. The average round time increased from \SI{9.01}{s} for 3 clients to \SI{10.30}{s} for 10 clients, and \SI{14.85}{s} for 30 clients. This growth is sub-linear, suggesting that the system can handle an increasing number of clients without prohibitive increases in round duration, at least within the tested range. The average per-client transaction time remained relatively stable or even decreased slightly (potentially due to averaging effects or Ganache behavior under load), indicating that individual blockchain submissions were not becoming significantly slower as client numbers increased.

\subsection{Cross-Dataset Comparison}
To confirm the consistency of PQC performance, we compared the results across the HAR, MNIST, and SVHN datasets for the 3-client, PQC-enabled, blockchain-integrated configuration. Figure~\ref{fig:cross_dataset} visualizes the training performance across these diverse datasets.

\begin{figure}[!htbp]
\centering
\includegraphics[width=\linewidth]{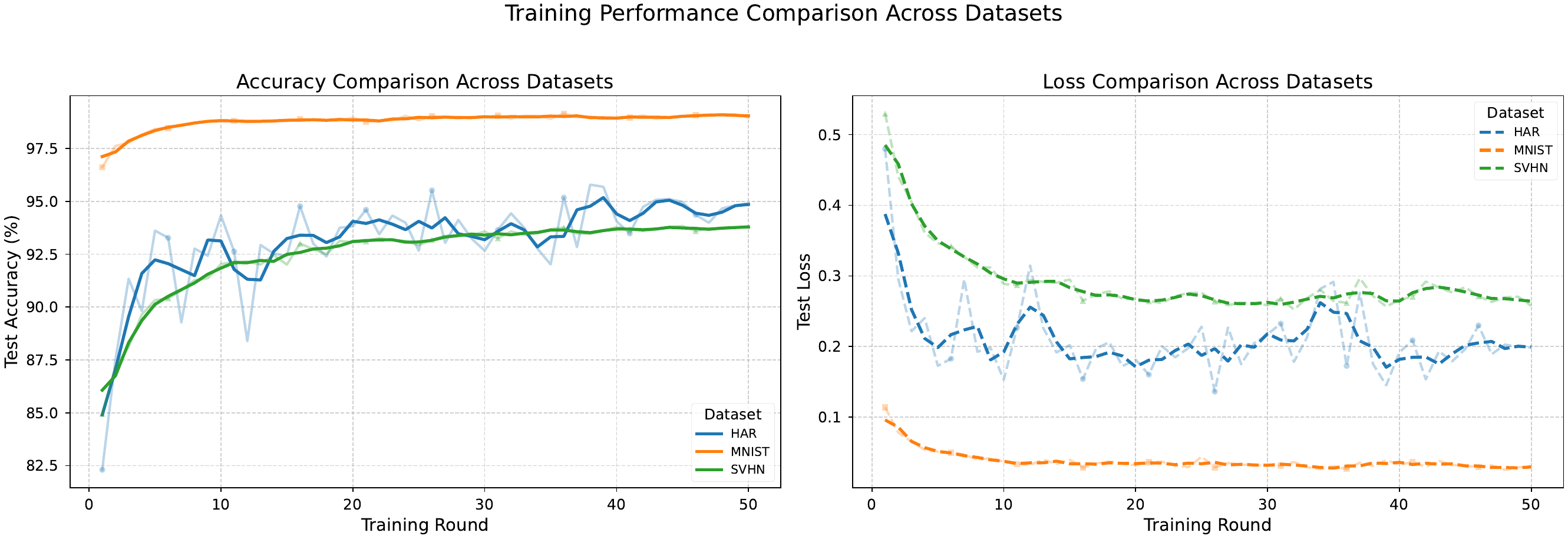}
\caption{Training performance comparison across HAR, MNIST, and SVHN datasets (PQC, 3 clients, blockchain enabled). Shows smoothed test accuracy (top) and loss (bottom), highlighting dataset-dependent learning curves but consistent PQC integration.}
\label{fig:cross_dataset}
\end{figure}

As expected, the core PQC cryptographic operation timings (Sign, Verify) remained consistent across datasets (see Table~\ref{tab:crypto_comp}), as these depend on the algorithm parameters, not the data being processed by the ML model. The differences in overall round times (\SI{3.05}{s} for HAR, \SI{9.01}{s} for MNIST, \SI{16.30}{s} for SVHN) and final model accuracy directly reflect the varying computational demands of the models and the inherent difficulty of the learning tasks associated with each dataset. Blockchain metrics like gas usage and transaction time were also primarily driven by the fixed-size PQC signature payload (\SI{3309}{B}) and associated verification complexity, rather than the dataset specifics.

\subsection{Impact of Blockchain Integration}
By comparing the MNIST PQC 3-client configuration with blockchain (\texttt{mnist-PQC-3c-BC}) and without (\texttt{mnist-PQC-3c-NoBC}, using a fixed delay), we can isolate the impact of the blockchain component. Integrating blockchain increased the average round time from \SI{6.88}{s} to \SI{9.01}{s}, representing an overhead of approximately 31\% attributed to transaction submission and confirmation in our simulated environment. There was a slight difference in final accuracy (99.41\% without BC vs. 99.19\% with BC), potentially due to timing variations affecting the training dynamics slightly. Most notably, the calculated cryptographic overhead ratio dropped dramatically from 2.31\% (NoBC) to 0.08\% (BC). This underscores that while blockchain adds latency, it makes the computational cost of PQC operations relatively insignificant within the end-to-end process.

\subsection{Discussion Summary}
The experimental results collectively demonstrate the practical feasibility of integrating ML-DSA-65 post-quantum signatures into a blockchain-based federated learning framework. While PQC introduces larger cryptographic artifacts and moderately slower signing/verification speeds compared to classical ECDSA (approximately 4x slower operations), the overall impact on the FL process is modest. Model accuracy is maintained, and the increase in round time is manageable, especially considering that local computation often dominates FL round durations. The primary trade-off lies in the increased gas costs for on-chain verification (roughly 9x higher than ECDSA in our estimations). However, our analysis critically shows that the computational overhead of PQC itself becomes negligible when placed within the context of blockchain transaction times. The system demonstrates good scalability with respect to the number of clients. These findings suggest that deploying PQS-BFL is a viable strategy for achieving long-term, quantum-resistant security in federated learning systems today, despite the higher blockchain resource consumption.

\section{Limitations and Future Work}
\label{sec:limitations}

While PQS-BFL demonstrates the feasibility of quantum-resistant blockchain-based FL, several limitations exist, opening avenues for future research.

\subsection{Limitations}
\begin{itemize}
    \item \textbf{Gas Costs:} The gas consumption for PQC update verification (avg. \num{1.7e6} per update) is significantly higher than for ECDSA. While manageable in our simulated environment, this could be a major cost barrier in large-scale deployments on public blockchains like Ethereum mainnet, especially with high gas prices.
    \item \textbf{Scalability Bottlenecks:} While round time scaled sub-linearly, potential bottlenecks could emerge at much larger scales (hundreds or thousands of clients) due to blockchain throughput limits (transactions per second), PQC verification load on nodes, or aggregator capacity.
    \item \textbf{Blockchain Storage:} Storing hashes or verification records for every client update across many rounds can lead to significant blockchain state growth over time, impacting node storage requirements and synchronization times.
    \item \textbf{Network Latency Simulation:} Using a local Ganache instance bypasses real-world network latency and variability, which could significantly impact actual transaction times and round durations.
    \item \textbf{PQC Verification Implementation:} Our smart contract simulation estimated PQC verification costs. A direct, gas-optimized Solidity implementation or standardized precompile for ML-DSA-65 verification would provide more accurate cost assessment and might differ from current estimates.
    \item \textbf{Simplified Aggregation:} The current framework assumes a straightforward aggregation process. More complex secure aggregation protocols, while enhancing privacy, could introduce additional cryptographic and communication overhead.
    \item \textbf{Limited PQC Schemes:} We focused solely on ML-DSA-65. Evaluating other NIST PQC signature finalists or candidates (e.g., Falcon, SPHINCS+) could reveal different performance trade-offs.
\end{itemize}

\subsection{Future Work}
Based on these limitations, future research directions include:
\begin{itemize}
    \item \textbf{Gas Optimization Techniques:} Implementing and evaluating advanced gas optimization strategies, such as batch verification (verifying multiple signatures in a single transaction), leveraging Layer-2 scaling solutions (e.g., rollups) to move verification off-chain while maintaining on-chain settlement, or exploring alternative blockchain platforms with lower transaction costs or built-in PQC support.
    \item \textbf{Optimized PQC Smart Contracts:} Developing highly optimized Solidity implementations for ML-DSA-65 verification or contributing to the standardization of PQC precompiles for major blockchain platforms.
    \item \textbf{Decentralized Aggregation and Verification:} Designing and evaluating protocols where aggregation and potentially even verification are performed by a decentralized network of nodes rather than a single aggregator, possibly using consensus mechanisms adapted for PQC.
    \item \textbf{Alternative PQC Signatures:} Benchmarking other quantum-resistant signature schemes (e.g., hash-based like SPHINCS+, or the more compact Falcon) within the PQS-BFL framework to compare performance trade-offs (e.g., signature size vs. computation time).
    \item \textbf{Integration with Privacy Techniques:} Combining PQS-BFL with privacy-enhancing technologies like differential privacy or secure aggregation protocols designed with PQC primitives to provide comprehensive security and privacy guarantees.
    \item \textbf{Real-World Network Evaluation:} Deploying and evaluating PQS-BFL on public blockchain testnets or private consortium blockchains to assess performance under realistic network conditions.
    \item \textbf{Adaptive Security Parameters:} Exploring mechanisms to dynamically adjust PQC security levels or parameters based on the specific requirements or risk profile of different FL rounds or client contributions.
\end{itemize}

\section{Conclusion}
\label{sec:conclusion}

The increasing sophistication of quantum computing necessitates a proactive transition towards quantum-resistant security measures in critical infrastructures, including collaborative machine learning systems like Federated Learning. This paper introduced PQS-BFL, a framework that addresses this challenge by integrating post-quantum digital signatures (ML-DSA-65) with blockchain-based verification for enhanced security and transparency in FL.

Our comprehensive experimental evaluation across diverse datasets (MNIST, SVHN, HAR) and system configurations demonstrates the practical feasibility of this approach. We showed that while ML-DSA-65 introduces larger cryptographic artifacts and moderately higher computational costs compared to classical ECDSA, its performance remains well within acceptable limits for practical deployment. Specifically, sub-millisecond signing and verification times were consistently observed. Crucially, when integrated with a blockchain, the overhead attributable to PQC operations becomes negligible (typically $<$0.2\%) compared to the overall transaction latency dominated by network and consensus mechanisms. This finding alleviates concerns that PQC might be computationally prohibitive for blockchain-based FL.

Furthermore, PQS-BFL maintains high model accuracy comparable to baselines without quantum-resistant security and exhibits favorable scalability characteristics, with round times growing sub-linearly as the number of clients increases. While PQC verification does incur higher gas costs on the blockchain compared to ECDSA, our analysis suggests this is a manageable trade-off for achieving long-term, quantum-resilient security.

PQS-BFL provides a robust and empirically validated blueprint for building next-generation federated learning systems capable of withstanding attacks from both classical and quantum adversaries. By releasing our framework and benchmarks, we aim to encourage further research and adoption of post-quantum security in decentralized AI applications, particularly in sensitive domains like healthcare where long-term data integrity and model authenticity are paramount. Future work will focus on optimizing blockchain costs, exploring alternative PQC schemes, and integrating enhanced privacy mechanisms to further strengthen the PQS-BFL framework.

\bibliographystyle{IEEEtran}
\bibliography{main}

\end{document}